\documentclass[reqno,11pt, centertags,draft]{amsart}
\usepackage{amssymb, geometry}
\usepackage{amsthm}
\usepackage{amsmath}
\usepackage[mathscr]{eucal}
\usepackage{latexsym}
\usepackage{xypic}
\usepackage{psfrag}

\numberwithin{equation}{section}
\newtheorem{theorem}{Theorem}[section]

\newtheorem{corollary}[theorem]{Corollary}

\newtheorem{lemma}[theorem]{Lemma}
\newtheorem*{conjecture}{Conjecture}
\newcommand{\tr}{\operatorname{tr}}

\newcommand{\sgn}{\operatorname{sgn}}
\newcommand{\Pf}{\operatorname{Pf}}
\newcommand{\sn}{\operatorname{sn}}
\newcommand{\Alt}{\operatorname{Alt}}

\newcommand{\ns}{\operatorname{ns}}
\newcommand{\cn}{\operatorname{cn}}
\newcommand{\dn}{\operatorname{dn}}
\newcommand{\ds}{\operatorname{ds}}
\newcommand{\cs}{\operatorname{cs}}
\newcommand{\even}{\operatorname{even}}
\newcommand{\odd}{\operatorname{odd}}
\newcommand{\alt}{\operatorname{alt}}
\newcommand{\row}{\operatorname{row}}

\newcommand{\sumtwo}[2]{\sum_{\substack{#1 \\ #2}}}

\newcommand{\Cliff}{\operatorname{Cliff}}

\newcommand{\xn}{\operatorname{xn}}
\newcommand{\yn}{\operatorname{yn}}
\newcommand{\nx}{\operatorname{nx}}
\newcommand{\x}{\operatorname{x}}
\newcommand{\y}{\operatorname{y}}

\theoremstyle{definition}

\begin{document}
\title{Periodic Ising Correlations}

\author{Grethe Hystad\\
   Department of Mathematics \\
   The University of Arizona \\
   617 N. Santa Rita Ave. \\
   Tucson, AZ 85721-0089 U.S.A.\\
   e-mail: ghystad@math.arizona.edu}

\maketitle

\begin{abstract}
In this paper, we first rework B. Kaufman's 1949 paper \cite{kaufman} by using representation theory. Our approach leads to a simpler and more direct way of deriving the spectrum of the transfer matrix for the finite periodic Ising model.
We then determine formulas for the spin correlation functions that depend on the matrix elements of the induced rotation associated with the spin operator in a basis of eigenvectors for the transfer matrix. The representation of the spin matrix elements is obtained by considering the spin operator as an intertwining map.
We exhibit the ``new'' elements $V_{+}$ and $V_{-}$ in the Bugrij-Lisovyy formula \cite{BL03} as part of a holomorphic factorization of the periodic and anti-periodic summability kernels on the spectral curve associated with the induced rotation for the transfer matrix.
\end{abstract}

\section{Introduction}

In 1944, L. Onsager \cite{onsager} determined the exact value of the specific heat as a function of temperature in the thermodynamic limit of the two-dimensional Ising model in the absence of an external magnetic field. He showed that the partition function can be approximated by the largest eigenvalue of the transfer matrix on the lattice. In 1949, B. Kaufman \cite{kaufman} simplified Onsager's calculation considerably by realizing the  transfer matrix for the finite periodic Ising model as an element in a spin representation of the orthogonal group. The eigenvalues of the transfer matrix could then be found from the angles of rotation of the rotation matrix. In order to calculate the degree of order, the transfer matrix was diagonalized by transformations obtained as spin representations of orthogonal rotations. However, the calculation of the determinant of these rotations (the proof of whether they are proper or improper) is not included in Kaufman's paper. 
(See line 33, page 1237 in \cite{kaufman}).
In this paper, we provide these calculations. 

There is a complex matrix $T_{z}(V)$ whose entries are rational functions of $z\in \mathbb{C}$ which completely determines the transfer matrix $V$ on the finite periodic lattice. The matrix $T_{z}(V)$ is called the induced rotation associated with the transfer matrix. The set of pairs $(\lambda,z)$ such that $\det(\lambda-T_{z}(V))=0$ is an elliptic curve $\mathcal{M}$ which is important for the spectral analysis of the transfer matrix. In particular, the map $\mathcal{M}\ni (\lambda,z)\to z\in\mathbb{P}^{1}$ is a two fold covering, and there are two cycles $\mathcal{M}_{\pm}$ on $\mathcal{M}$ which cover the circle $\mathbb{S}^{1}=\{z:|z|=1\}$. On the cycle $\mathcal{M}_{+}$ we have $\lambda<1$, and on the cycle $ \mathcal{M}_{-}$, we have $\lambda>1$. 
Just which points $z_{j}\in \mathbb{S}^{1}$ are relevant for the spectral analysis depend on the boundary conditions for the model. For spin periodic boundary conditions on the lattice, the $(2M+1)^{th}$ roots of unity, $z^{2M+1}=1$, are relevant as are the $(2M+1)^{th}$ roots of $-1$, $z^{2M+1}=-1$. We will refer to these two finite sets as the periodic spectrum $\Sigma_{P}$ and the anti-periodic spectrum $\Sigma_{A}$.
In the infinite-volume limit all the points $z\in \mathbb{S}^{1}$ are relevant.
Kaufman \cite{kaufman} showed that the 
space in which the transfer matrix acts, can be divided into two invariant subspaces, which
we will denote by $(U=1)$ and $(U=-1)$. Here $U$ is a product of the basis elements of the finite sequence space, $W:=l^{2}(-M,...,M,\mathbb{C}^{2})$, which are certain representations of the Clifford relations. The expression $(U=\pm 1)$ is the short-hand notation for the $\pm 1$ eigenspaces of $U$. More specifically, Kaufman showed that the transfer matrix $V$ can be written as the direct sum $V=V^{A}\oplus V^{P}$, where $V^{A}=V|_{(U=1)}$ and $V^{P}=V|_{(U=-1)}$. Here we have introduced the letters $A$ and $P$ to refer to the restriction to the anti-periodic and periodic spectrum respectively. We rework Kaufman's paper \cite{kaufman} by proving that $(U=1)$ is unitarily equivalent to the even tensor algebra of a subspace of $W$ in the anti-periodic Fourier representation while $(U=-1)$ is unitarily equivalent to the even tensor algebra of a subspace of $W$ in the periodic Fourier representation. This result is summarized in Theorem \ref{U=1} and is the main result of this paper. In the anti-periodic Fock representation, we can choose a representation of the transfer matrix $V^{A}$ such that the vacuum vector $0_{A}$ is an eigenvector associated with its largest eigenvalue. In a similar fashion, we can choose a representation of the transfer matrix $V^{P}$ in the periodic Fock representation such that the vacuum vector $0_{P}$ is an eigenvector associated with its largest eigenvalue. By using these representations of the transfer matrix, we can easily compute its full spectrum. We believe that this approach is much more direct and simpler than the method employed in Kaufman's paper. By applying this representation theoretic approach, we can also determine formulas for the spin matrix elements, in a basis of eigenvectors for the transfer matrix, that depend on the matrix elements of the induced rotation associated with the spin operator. The spin matrix elements will be calculated by regarding the spin operator as an intertwining map for the periodic and anti-periodic Fock representations.   

The paper is organized as follows.
In Section \ref{Ising}, we describe the transfer matrix approach to the Ising model as given in \cite{kaufman}, and we provide the formula for the induced rotation associated with the transfer matrix. In Section \ref{TheSpectrum}, we prove Theorem \ref{U=1} which we described above, and apply it to derive the spectrum of the transfer matrix in Theorem \ref{spectrum}.  
In Section \ref{Inducedrotationforthespinoperator}, we compute the matrix elements of the induced rotation associated with the spin operator by restricting the spin operator to be a map from the Fourier space with periodic boundary conditions to the Fourier space with anti-periodic boundary conditions. This result is given in Theorem \ref{smatrix}. Let us denote the matrix of the induced rotation for the spin operator by
$$s:=\left(\begin{array}{lll}
A&B\\
C&D
\end{array}\right),$$ where $A,B,C,D$ are matrix elements for $s$ in a polarization of $W$.
The correlation functions on the cylinder and the torus can be evaluated in terms of spin matrix elements in an orthonormal basis of eigenvectors for the transfer matrix. We recall this calculation in Section \ref{Twopoint}. Subsequently, in Section \ref{Spinmatrixelements}, Theorem \ref{spinmatrix5}, we provide Pfaffian formulas for the spin correlation functions that depend on the inverse, $D^{-1}$. To invert $D$ is an extremely difficult task, and we only have a formula for it in terms of a Bugrij-Lisovyy conjecture for the spin matrix elements. 
Bugrij and Lisovyy  proposed the explicit formulas for the spin matrix
elements on the finite periodic lattice for the isotropic case \cite{BL03} and
for the anisotropic case \cite{BL04}. We include these formulas in Section \ref{BLformula}. A proof of these formulas was given in \cite{GIPST07} and \cite{GIPST08} but it is rather complicated.  In Section \ref{Pfaffianformalism}, we express the proposed formulas for the spin matrix elements in terms of a product of Jacobian elliptic functions in the uniformization parameter of the spectral curve.
In Section \ref{Numerical}, we discuss the numerical comparison of our $D^{-1}$ matrix elements with the corresponding term in the Bugrij-Lisovyy formula.
In Section \ref{spectral2}, we turn to a calculation that shows the role played by the ``new ''elements $V_{\pm}$
in the Bugrij-Lisovyy formula in a factorization of the ratio of summability kernels on the
spectral curve. We summarize this result in Theorem \ref{V+V_lemma}. 
In Appendix \ref{Appendix A}, we introduce the Fock representation of the Clifford Algebra of $W$ as given in \cite{palmer1}.
In Appendix \ref{Appendix B}, we introduce a Berezin integral representation for the matrix elements of the Fock representation of an element in the Clifford group. We express the kernel of this operator as an exponential of a skew symmetric matrix whose matrix elements depend on the matrix elements of its induced rotation. We give this result in Lemma \ref{LemmaB} and Theorem \ref{gPfaffian}.
Finally, in Appendix \ref{Appendix C}, we provide the formula for the spectral curve associated with the induced rotation for the transfer matrix. We recall some properties of the Jacobian Elliptic functions which are involved in
the holomorphic factorization of the summability kernels on the spectral curve and in the product formula for the spin matrix elements.
\section{Transfer Matrix}\label{Ising}
We begin this section by recounting the analysis in the 1949 paper of Bruria Kaufman \cite{kaufman} which shows that the transfer matrix for the finite periodic Ising model can be expressed as an element in a spin representation of the orthogonal group. We then continue with a calculation of the induced rotation associated with the transfer matrix.

For positive integers, $M$ and $N$, consider the set $I_{M}=\{-M,-M+1,...,M\}$ and the finite lattice
$\Lambda:=I_{M}\times I_{N}$.
Let each vertex be assigned a spin value of $+1$ (spin up) or $-1$ (spin down).
A configuration $\sigma$ is a particular assignment of spin values to the vertices, i.e. a spin configuration is a map,
$$\sigma:\Lambda\to \{+1,-1\}.$$
In this paper, $\sigma$ satisfies the periodic boundary conditions $\sigma(-M,j)=\sigma(M+1,j)$ and $\sigma(j,-N)=\sigma(j,N+1)$ for all $j$. Each spin interacts ferromagnetically with its nearest neighbors; in a configuration $\sigma$ the interaction energy in the absence of an external magnetic field is defined by
$$E_{\Lambda}(\sigma)=-\sum_{(i,j) \in \Lambda,|i-j|=1}J_{ij}\sigma_{i}\sigma_{j}.$$ The sum is over all nearest-neighbors $(i,j)$ in $\Lambda$ with real valued interaction constants, $J_{ij}=J_{1}>0$, if the sites are horizontally separated, and $J_{ij}=J_{2}>0$, if the sites are vertically separated. 
The probability of a given configuration $\sigma$ is proportional to the Boltzmann weight,
$$w(\sigma):=\exp{\bigg(-\tfrac{E_{\Lambda}(\sigma)}{k_{B}T}\bigg)},$$ where $T$ is the temperature and $k_{B}$ is the Boltzmann constant.
The total weight is given in terms of the partition function,
$$Z_{\Lambda}=\sum_{\sigma \in \Omega_{\Lambda}}w(\sigma),$$ where $\Omega_{\Lambda}$ is the set of all possible configurations on $\Lambda.$
The correlation functions are the expected values of products of spin variables at sites $j_{1},...,j_{n}$ in $\Lambda$,
$$\langle {\sigma_{j_{1}}}\cdot \cdot \cdot \cdot {\sigma_{j_{n}}}\rangle_{\Lambda}=\frac{1}{Z_{\Lambda}}\sum_{\sigma \in \Omega_{\Lambda}}{\sigma_{j_{1}}}\cdot \cdot \cdot \cdot {\sigma_{j_{n}}}w(\sigma).$$
Let $\Omega_{\Lambda}(\row)$ denote the space of configurations of a row. An $i^{th}$ row configuration is a map,
$$\sigma^{i}:I_{M}\to \{+1,-1\}.$$ For $i=-N,...,N$, we denote a collection of configurations
 $\sigma^{i} \in \Omega_{\Lambda}(\row)$ as $\sigma_{j}^{i}:=\sigma_{ij}.$
Thus, the spin variable $\sigma_{ij}$ is located at the site $j$ in the $i^{th}$ row. For $\sigma, \tau \in \Omega_{\Lambda}(\row)$ we define the $2^{2M+1}$ dimensional matrices,
\begin{align}\label{V1}
V_{1}(\sigma):  = \exp(\sum_{j=-M}^{M}{\mathcal{K}_{1}\sigma_{j}\sigma_{j+1}}), \quad V_{2}(\sigma,\tau):  =\exp(\sum_{j=-M}^{M}\mathcal{K}_{2}\sigma_{j}\tau_{j}),\end{align} where
$\mathcal{K}_{l}:=\frac{J_{l}}{k_{B}T}$ for $l=1,2.$
Kaufman \cite{kaufman} showed that the partition function can be written as the trace of the $2N+1$ power of the transfer matrix $V_{1}V_{2}$,
\begin{align}\label{part}
Z_{\Lambda}
&  =\tr[(V_{1}V_{2})^{2N+1}].
\end{align}
Introduce the matrices,
$$\sigma = \left( \begin{array}{cc}
1 & 0 \\
0 & -1  \\
 \end{array} \right), \quad
C = \left( \begin{array}{cc}
0 & 1 \\
1 & 0  \\
 \end{array} \right),\quad 
I = \left( \begin{array}{cc}
1 & 0 \\
0 & 1  \\
 \end{array} \right),$$
and define
\begin{align}\label{sigma'}
\sigma_{j}:=  \underbrace{I \otimes \cdot \cdot \cdot \otimes I}_{M+j} \otimes \sigma \otimes I \otimes \cdot \cdot \cdot  \otimes I
\end{align}
\begin{align*}
C_{j} : =  \underbrace{I \otimes \cdot \cdot \cdot  \otimes I}_{M+j}\otimes C \otimes I \otimes \cdot \cdot \cdot  \otimes I
\end{align*}
which act on the tensor product, $$\bigotimes_{j=-M}^{M}\mathbb{C}^{2}_{j}:=\mathbb{C}^{2}_{-M} \otimes \cdot \cdot \cdot \otimes \mathbb{C}^{2}_{M}=\mathbb{C}^{2(2M+1)},$$ where $\mathbb{C}^{2}_{j}=\mathbb{C}^{2}$ for each $j.$ 
One can pick a basis for this tensor product space such that the action of the spin operator $\sigma_{j}$ for $\sigma\in \Omega_{A}(\row)$ is given by (\ref{sigma'}) \cite{palmer1}. 
For $-M\leq k \leq M$, define
\begin{equation*}
\begin{array}{lccc}
p_{k}:=\underbrace{C \otimes \cdot \cdot \cdot \otimes C}_{M+k}\otimes \sigma \otimes I \otimes \cdot \cdot \cdot \otimes I,\\
q_{k}:=\underbrace{C \otimes \cdot \cdot \cdot \otimes C}_{M+k}\otimes -i\sigma C \otimes I \otimes \cdot \cdot \cdot \otimes I
\end{array}
\end{equation*}  
which satisfy the generator relations for the Clifford algebra,
\begin{align}\label{CRelations}
p_{k}p_{l}+p_{l}p_{k}=2\delta_{kl},\quad q_{k}q_{l}+q_{l}q_{k}=2\delta_{kl}, \quad p_{k}q_{l}+q_{l}p_{k}=0.
\end{align} 
Introduce the finite sequence space, $W:=l^{2}(I_{M},\mathbb{C}^{2})$, with an orthonormal basis \newline
$\{\tfrac{q_{k}}{\sqrt{2}},\tfrac{p_{k}}{\sqrt{2}}\}_{k=-M}^{M}$ such that for an element $v \in W,$ we have
\begin{align*}
v=\tfrac{1}{\sqrt{2}}\sum_{k=-M}^{M}x_{k}(v)q_{k}+y_{k}(v)p_{k}.
\end{align*}
We define the distinguished nondegenerate complex bilinear form $(\cdot,\cdot)$ on $W$ by
\begin{align}(u,v)\label{bilinear}=\sum_{k=-M}^{M}x_{k}(u)x_{k}(v)+y_{k}(u)y_{k}(v).\end{align}
The Hermitian inner product on $W$ is given by $\langle u,v\rangle=(\overline{u},v)$, where the conjugation is defined as
$$\overline{u}=\tfrac{1}{\sqrt{2}}\sum_{k=-M}^{M}\overline{x}_{k}q_{k}+\overline{y}_{k}p_{k}.$$ The representations \{$q_{k}, p_{k}\}$ are generators of an irreducible $*$-representation of the Clifford algebra, $\Cliff(W)$, on $\mathbb{C}^{2(2M+1)}$ \cite{palmer1}.  
Now define the operator,
$$U:=\prod_{k=-M}^{M}C_{k}=\prod_{k=-M}^{M}ip_{k}q_{k},$$ which plays a central role in Kaufman's analysis. 
It is shown in \cite{kaufman} that the transfer matrices $V_{1}$ and $V_{2}$ in (\ref{V1}) can be written\begin{align*}
\begin{array}{clcc}
V_{1}&=(\prod_{j=-M}^{M-1}\exp{(-i\mathcal{K}_{1}p_{j+1}q_{j})})\exp{(i\mathcal{K}_{1}p_{-M}q_{M}U)},\\
V_{2}&=\prod_{j=-M}^{M}\exp{(i\mathcal{K}^{*}_{2}p_{j}q_{j})},
\end{array}
\end{align*}
where
the dual interaction constant $\mathcal{K}^{*}_{2}$ is defined by the relation $$\sinh(2\mathcal{K}^{*}_{2})\sinh(2\mathcal{K}_{2})=1.$$
It is easy to check that $U^{2}=1$ and that $U$ commutes with even elements in the Clifford algebra. From this it follows immediately that $V_{1}$ and $V_{2}$ leave invariant the $\pm 1$ eigenspaces for $U$. 
Let
$(U=\pm1)$ denote the $\pm 1$ eigenspaces for $U$.
Define $$V:=V_{1}V_{2}$$ and
\begin{align}\label{VAVP2}
V^{A}: & = \frac{1}{2}(I+U)\bigg(\prod_{j=-M}^{M}\exp{i\mathcal{K}^{*}_{2}}p_{j}q_{j}\bigg)\bigg(\prod_{j=-M}^{M-1}\exp{(-i\mathcal{K}_{1}p_{j+1}q_{j})}\bigg)\exp{(i\mathcal{K}_{1}p_{-M}q_{M})},\\
V^{P}: &\label{VAVP3} = \frac{1}{2}(I-U)\bigg(\prod_{j=-M}^{M}\exp{(i\mathcal{K}^{*}_{2}p_{j}q_{j})}\bigg)\bigg(\prod_{j=-M}^{M}\exp{(-i\mathcal{K}_{1}p_{j+1}q_{j})}\bigg).
\end{align}
The letters $P$ and $A$ refer to periodic and anti-periodic respectively, and the reason for their appearance will be clear shortly.
A central result in \cite{kaufman} is that the transfer matrix $V$ can be written as the direct sum,
$$V=V^{A}\oplus V^{P},$$ where $$V^{A}=V|_{(U=1)}\quad \mbox{and}\quad V^{P}=V|_{(U=-1)}.$$ 
The exponential factors in $V^{A}$ and $V^{P}$ are elements in a spin representation of the orthogonal group \cite{kaufman}. 

A linear transformation $V$ on $\displaystyle{\otimes_{j=-M}^{M}\mathbb{C}^{2}_{j}}$ is an element of the Clifford group if there exists a linear transformation,
\label{def2}$$T(V):W \to W,$$ such that for all $w \in W \subseteq \Cliff(W)$, we have
$$VwV^{-1}=T(V)w.$$ The operator $T(V)$ is called the induced rotation associated with $V$.
The induced rotation $T(V)$ is complex orthogonal with respect to the bilinear form $(\cdot, \cdot)$ and it determines $V$ up to a scalar multiple \cite{palmer1}.
For $j=1,2$ introduce the short-hand notation
$$c_{j}:=\cosh(2\mathcal{K}_{j})\quad\mbox{and} \quad s_{j}:=\sinh(2\mathcal{K}_{j})$$
for the hyperbolic parametrization of the Boltzmann weights.
We write $$c^{*}_{j}:=\cosh(2\mathcal{K}^{*}_{j})=\frac{c_{j}}{s_{j}} \quad \mbox{and} \quad s^{*}_{j}:=\sinh(2\mathcal{K}^{*}_{j})=\frac{1}{s_{j}}$$
for the dual interactions. 
Adapting the technique introduced in \cite{palmer1}, we compute the induced rotation associated with the transfer matrix in $(x,y)$ coordinates,
\begin{align*}T(V_{2})\left( \begin{array}{c}
x  \\
y   \\
 \end{array} \right)=\left( \begin{array}{cc}
c^{*}_{2}& -is^{*}_{2} \\
is^{*}_{2}& c^{*}_{2}  \\
 \end{array} \right)\left( \begin{array}{c}
x  \\
y  \\
 \end{array} \right),\end{align*}
\begin{align*}
T(V_{1})\left( \begin{array}{c}
x  \\
y   \\
 \end{array} \right)=\left( \begin{array}{cc}
c_{1}& is_{1}z^{-1}  \\
-is_{1}z& c_{1}   \\
 \end{array} \right)\left( \begin{array}{c}
x  \\
y   \\
 \end{array} \right).
 \end{align*}
The second equation is understood as follows. The operator $z$ acts on a finite sequence $x$ by sending $x_{k}$ to $x_{k-1}$. On the invariant subspace for $V_{1}$, where $(U=-1)$, the sequence is extended to be $2M+1$ periodic. On the invariant subspace for $V_{1}$, where $(U=1)$, the sequence is extended to be $2M+1$ anti-periodic.
For $k=-M,...,M$, introduce 
$$\theta^{P}_{k}:=\frac{2\pi k}{2M+1}, \quad \theta^{A}_{k}:=\frac{2\pi (k+\tfrac{1}{2})}{2M+1}$$ and
$$z_{P}:=z_{P}(k)=e^{i\theta^{P}_{k}},\quad z_{A}:=z_{A}(k)=e^{i\theta^{A}_{k}}.$$
Define the two sets, $\Sigma_{P}=\{z\in \mathbb{C}:z^{2M+1}=1\}$ and $\Sigma_{A}=\{z\in \mathbb{C}:z^{2M+1}=-1\}$, which we refer to as the periodic and anti-periodic spectrum respectively.
By specializing the finite Fourier series,
$$x(z)=\frac{1}{\sqrt{2M+1}}\sum_{k=-M}^{M}x_{k}z^{k},$$ to the periodic spectral points, $z\in \Sigma_{P}$, on the subspace $(U=-1)$ 
or to the anti-periodic spectral points, $z\in \Sigma_{A}$, on the subspace $(U=1)$, we obtain that $T(V_{1})$ acts as a multiplication operator,
\begin{align*}
T(V_{1})\left( \begin{array}{c}
x(z)  \\
y(z)   \\
 \end{array} \right)=\left( \begin{array}{cc}
c_{1}& is_{1}z^{-1}  \\
-is_{1}z& c_{1}   \\
 \end{array} \right)\left( \begin{array}{c}
x(z)  \\
y(z)   \\
 \end{array} \right).
 \end{align*}
Define the induced rotation for the symmetrized transfer matrix,
\begin{align*}
T(V):=T(V^{\tfrac{1}{2}}_{2})T(V_{1})T(V^{\tfrac{1}{2}}_{2}).
\end{align*}
Then for temperatures below the critical temperature $T_{C}$ we find, 
\begin{align}\label{Matrix}
T(V)=
\left( \begin{array}{cc}
\cosh \gamma(z)& w(z)\sinh \gamma(z)  \\
\overline{w}(z)\sinh \gamma(z)& \cosh \gamma(z) \\
 \end{array} \right),
 \end{align} where $\gamma(z)$ is defined as the positive root of
\begin{align}\label{coshgamma}\cosh \gamma(z)=c^{*}_{2}c_{1}-s_{2}^{*}s_{1}\big(\tfrac{z+z^{-1}}{2}\big),\end{align} and where  \begin{align*}
w(z)=iz^{-1}\frac{\mathcal{A}_{1}(z)\mathcal{A}_{2}(z)}{\mathcal{A}_{1}(z^{-1})\mathcal{A}_{2}(z^{-1})}.
\end{align*}  
Here $\mathcal{A}_{j}(z)$ is defined by $$\mathcal{A}_{j}(z):=\sqrt{\alpha_{j}-z}\quad \mbox{for \quad j=1, 2},$$
where
$$\alpha_{1}:=(c^{*}_{1}-s^{*}_{1})(c_{2}+s_{2})=e^{2(\mathcal{K}_{2}-\mathcal{K}^{*}_{1})},$$
$$\alpha_{2}:=(c^{*}_{1}+s^{*}_{1})(c_{2}+s_{2})=e^{2(\mathcal{K}_{2}+\mathcal{K}^{*}_{1})}.$$ 
For $T<T_{C}$, we have that $1 < \alpha_{1}< \alpha_{2}$, and the square root is chosen to lie in the right half plane with positive real part. We observe that
$z\to (\mathcal{A}_{j}(z))^{\pm}$ is analytic for $z$ in a neighborhood of the unit disc $(|z|<\alpha_{1})$ while
$z\to (\mathcal{A}_{j}(z^{-1}))^{\pm}$ is analytic in a neighborhood of the exterior of the unit disc $(|z|>\alpha^{-1}_{1})$.

\section{The Spectrum of the Transfer Matrix}\label{TheSpectrum}
Kaufman \cite{kaufman} showed that the eigenvalues of the transfer matrix can be found from the angles of rotation of the rotation matrix by realizing the $2^{2M+1}$ dimensional transfer matrix as a spin representation of the $(2M+1)$ dimensional rotation matrix. In this section we compute the eigenvalues for the transfer matrices $V^{A}$ and $V^{P}$ by using representation theory. This approach is simpler and more direct than Kaufman's method.

After Fourier transform, the Hermitian inner product on $W$ becomes
$$\sum_{k=-M}^{M}\bar{x}(z_{k})x'(z_{k})+ \bar{y}(z_{k})y'(z_{k})$$ and the complex bilinear form becomes 
\begin{align}\label{bilinearf}
\sum_{k=-M}^{M}x(z_{k})x'(z^{-1}_{k})+ y(z_{k})y'(z^{-1}_{k}),
\end{align} where the conjugation is given by $x(z)\mapsto \bar{x}(z^{-1})$. 
Let $T^{A}$ denote the induced rotation associated with $V^{A}$, where $V^{A}$ is given in (\ref{VAVP2}). Let $T^{P}$ denote the induced rotation associated with $V^{P}$, where $V^{P}$ is given in (\ref{VAVP3}). Both $T^{A}$ and $T^{P}$ have positive real spectrum which does not contain $1$.
We are interested in the isotropic splittings, $$W=W^{A}_{+}\oplus W^{A}_{-} \quad \mbox{and}\quad
W=W^{P}_{+}\oplus W^{P}_{-},$$ where $W^{A}_{+}$ and $W^{P}_{+}$ are the span of all eigenvectors of $T^{A}$ and $T^{P}$ respectively associated with the eigenvalues between $0$ and $1$, and 
$W^{A}_{-}$ and $W^{P}_{-}$ are the span of all eigenvectors of $T^{A}$ and $T^{P}$ respectively associated with the eigenvalues greater than $1$. 
Define
\begin{align*}
{a}(z):=\sqrt{\frac{\mathcal{A}_{1}(z)\mathcal{A}_{2}(z)}{\mathcal{A}_{1}(z^{-1})\mathcal{A}_{2}(z^{-1})}} \quad
\end{align*} which is normalized such that $a(1)>0$.
The eigenvalues of $T^{P}$ are given by $e^{-\gamma(z_{P}(k))}$ and $e^{\gamma(z_{P}(k))}$ for $k=-M,...,M$, where $\gamma(z)$ is defined as the positive root of (\ref{coshgamma}). The corresponding eigenvectors are
\begin{align}\label{eigenP1}e^{P}_{+,k}(z):=\frac{1}{\sqrt{2}}\left( \begin{array}{cc}
a(z)  \\
iza(z)^{-1}   \\
 \end{array} \right)\delta_{z_{P}(k)}(z)
\end{align} and
\begin{align}\label{eigenP2}e^{P}_{-,k}(z):=\frac{1}{\sqrt{2}}\left( \begin{array}{cc}
a(z)  \\
-iza(z)^{-1}  \\
 \end{array} \right)\delta_{z_{P}(-k)}(z)
 \end{align} respectively for $z\in \Sigma_{P}$. Here \[ \delta_{z_{P}(k)}(z) = \left\{ \begin{array}{ll}
         1 & \mbox{if $z =z_{P}(k)$};\\
           0 & \mbox{if $z \neq z_{P}(k)$}.\end{array} \right. \]
Then $\{e^{P}_{+,k}, e^{P}_{-,k}\}_{k=-M}^{M}$ is an orthonormal basis for $W$ with respect to the Hermitian inner product. We also have $(e^{P}_{+,k}, e^{P}_{-,l})=\delta_{kl}$, so $\{e^{P}_{+,k}\}_{k=-M}^{M}$ and $\{e^{P}_{-,k}\}_{k=-M}^{M}$ are dual basis for $W^{P}_{+}$ and $W^{P}_{-}$ with respect to the complex bilinear form.
The eigenvalues and eigenvectors, $e^{A}_{\pm,k}(z)$, of $T^{A}$ are defined in the exact same way for $z\in \Sigma_{A}$. From (\ref{Matrix}) we see that $W^{A}_{\pm}$ are the $\pm 1$ eigenspaces of the polarization given by the multiplication operators in the anti-periodic Fourier representation,
\begin{align}\label{QQ}
Q(z_{A})=-\left( \begin{array}{cc}
0 & w(z_{A}) \\
\overline{w}(z_{A}) &0  \\
 \end{array} \right),
\end{align}
and $W^{P}_{\pm}$ are the $\pm 1$ eigenspaces of the polarization given by the multiplication operators in the periodic Fourier representation,
\begin{align}\label{QQ2}
Q(z_{P})=-\left( \begin{array}{cc}
0 & w(z_{P}) \\
\overline{w}(z_{P}) &0  \\
 \end{array}\right).
 \end{align}
 Let $Q^{A}$ denote multiplication by $Q(z_{A})$, and let $Q^{P}$ denote multiplication by $Q(z_{P})$.
Define
$$Q^{A}_{\pm}:=\tfrac{1}{2}(I\pm Q^{A})\quad \mbox{and} \quad Q^{P}_{\pm}:=\tfrac{1}{2}(I\pm Q^{P}).$$ Then
$$W^{A}_{\pm}=Q^{A}_{\pm}W\quad \mbox{and}\quad W^{P}_{\pm}=Q^{P}_{\pm}W$$ so $Q^{A}_{\pm}$ and $Q^{P}_{\pm}$ are orthogonal projections on $W^{A}_{\pm}$ and $W^{P}_{\pm}$ respectively.
Define the alternating tensor algebras,
$$\Alt_{\even}(W_{+}):=\bigoplus_{0\leq 2k\leq n}\Alt^{2k}(W_{+})\quad \mbox{and} \quad \Alt_{\odd}(W_{+}):=\bigoplus_{0\leq 2k+1\leq n}\Alt^{2k+1}(W_{+}),$$ where $\Alt^{k}(W_{+})$ is the space of alternating $k$ tensors over $W_{+}$, $\Alt^{0}(W_{+})=\mathbb{C}$ and $n=\dim(W_{+})$.
Let $F^{P}$ and $F^{A}$ denote the Fock representations associated with the Clifford relations acting on $\Alt(W^{P}_{+})$ and $\Alt(W^{A}_{+})$ respectively.
We have for $x\in W$,
\begin{align}\label{Fock7}
F^{P}(x)=c(x_{+})+a(x_{-}),
\end{align} where
$x_{\pm}=Q^{P}_{\pm}x$, and $c(\cdot)$ and $a(\cdot)$ are creation and annihilation operators as defined in (\ref{creationop}) and (\ref{annihilationop}). 
Recall that $V=V^{A}\oplus V^{P}$, where $V^{A}=V_{|(U=1)}$ and $V^{P}=V_{|(U=-1)}$, and where $\displaystyle{U=\prod_{k=-M}^{M}ip_{k}q_{k}.}$
We will prove the following theorem.
\begin{theorem}\label{U=1}
Consider the isotropic splittings $W=W^{A}_{+}\oplus W^{A}_{-}$ and 
$W=W^{P}_{+}\oplus W^{P}_{-}$ associated with the polarizations defined in (\ref{QQ}) and (\ref{QQ2}).
Let $(U=\pm 1)$ denote the $\pm 1$ eigenspaces for $U$, where  $\displaystyle{U=\prod_{k=-M}^{M}ip_{k}q_{k}}$.
Then
$$\Alt_{\even}(W^{A}_{+})\simeq (U=1), \quad \Alt_{\even}(W^{P}_{+})\simeq (U=-1),$$
$$\Alt_{\odd}(W^{A}_{+})\simeq (U=-1), \quad \Alt_{\odd}(W^{P}_{+})\simeq (U=1),$$
where $\simeq$ denotes unitarily equivalent.
\end{theorem}
\begin{proof}
The strategy of the proof is as follows.
We first make a change of basis $\{\tfrac{{q_{k}}}{\sqrt{2}},\tfrac{p_{k}}{\sqrt{2}}\}_{k=-M}^{M}$ to a basis
$\big\{\tfrac{1}{\sqrt{2}}q(z(k)),\tfrac{1}{\sqrt{2}}p(z(k))\big\}_{k=-M}^{M},$ where $q(z(k))$ and $p(z(k))$ are real with respect to the conjugation, $v(z)\mapsto \bar{v}(z^{-1})$, for $z$ in $\Sigma_{P}$ or in $\Sigma_{A}$. The elements $q(z(k))$ and $p(z(k))$ will be defined below.
Then define the ``volume elements'', $$U^{A}:=\prod_{k=-M}^{M}ip(z_{A}(k))q(z_{A}(k))\quad \mbox{and}\quad U^{P}:=\prod_{k=-M}^{M}ip(z_{P}(k))q(z_{P}(k)),$$ in the Clifford algebra.
Let $\mathcal{N}^{P}$ and $\mathcal{N}^{A}$ denote the number operators in $\Alt(W^{P}_{+})$ and $\Alt(W^{A}_{+})$ respectively, which are defined as \begin{align}\label{NPA}\mathcal{N}^{P}|\Alt^{k}(W^{P}_{+})=k\quad \mbox{and}\quad \mathcal{N}^{A}|\Alt^{k}(W^{A}_{+})=k.\end{align}
We show below that \begin{align}
\label{NA}U^{P}=(-1)^{\mathcal{N}^{P}}\quad \mbox{and} \quad U^{A}=(-1)^{\mathcal{N}^{A}}.
\end{align}
Denote the linear transformation that sends $\tfrac{q_{k}}{\sqrt{2}}$ to $\tfrac{1}{\sqrt{2}}q(z_{P}(k))$ and $\tfrac{p_{k}}{\sqrt{2}}$ to $\tfrac{1}{\sqrt{2}}p(z_{P}(k))$ by $R^{P}$, and the linear transformation that sends $\tfrac{q_{k}}{\sqrt{2}}$ to $\tfrac{1}{\sqrt{2}}q(z_{A}(k))$ and $\tfrac{p_{k}}{\sqrt{2}}$ to $\tfrac{1}{\sqrt{2}}p(z_{A}(k))$ by $R^{A}$. The transformations, $R^{P}$ and $R^{A}$, are real and orthogonal, and it follows from page 153 of \cite{BS96} that   
they induce an automorphism of the Clifford algebra such that
\begin{align}\label{SPSA}
U=\det(R^{P})U^{P}\quad
\mbox{and}\quad
U=\det(R^{A})U^{A}.
\end{align}
We show below that
\begin{align}\label{detRARP3}
\det R^{P}=-1 \quad \mbox{and}\quad \det R^{A}=1.
\end{align} 
Then it follows from (\ref{NA}), (\ref{SPSA}) and (\ref{detRARP3}) that
\begin{align}\label{U11}U=-(-1)^{\mathcal{N}^{P}} \quad \mbox{on} \quad \Alt(W^{P}_{+})\end{align} and \begin{align}\label{U22}U=(-1)^{\mathcal{N}^{A}}\quad \mbox{on} \quad \Alt(W^{A}_{+}).\end{align}
Since $(-1)^{\mathcal{N}_{P}}=1$ on $\Alt_{\even}(W^{P}_{+})$ and $(-1)^{\mathcal{N}_{A}}=1$ on $\Alt_{\even}(W^{A}_{+})$, we have from (\ref{U11}) and (\ref{U22}) that
$$\Alt_{\even}(W^{P}_{+})\simeq (U=-1)\quad\mbox{ and} \quad \Alt_{\even}(W^{A}_{+})\simeq(U=1).$$
Subsequently, it follows that
$$\Alt_{\odd}(W^{P}_{+})\simeq (U=1)\quad \mbox{and} \quad
\Alt_{\odd}(W^{A}_{+})\simeq (U=-1).$$
We now return to the introduction of the basis elements, $\{\tfrac{1}{\sqrt{2}}q(z_{P}(k)),\tfrac{1}{\sqrt{2}}p(z_{P}(k))\}_{k=-M}^{M}$ and then proceed to prove the first identity in (\ref{NA}).
In the polarization, $W=W^{P}_{+}\oplus W^{P}_{-}$, define the creation operators as $$a_{k}^{P*}:=c(e^{P}_{+,k})$$ and the annihilation operators  as $$a_{k}^{P}:=a(e^{P}_{-,k}),$$ where $c(\cdot)$ and $a(\cdot)$ are defined in (\ref{creationop}) and (\ref{annihilationop}), and
$\{e^{P}_{\pm,k}\}$  are the eigenvectors for the induced rotation $T^{P}$ as defined in (\ref{eigenP1}) and (\ref{eigenP2}).
In the $z_{P}$ coordinates we define the basis elements for $W$ as
\begin{align}\label{pzP}
p(z_{P}(k)):=(a^{P*}_{k}+a^{P}_{k})\quad\mbox{and} \quad
q(z_{P}(k)):=i(a^{P}_{k}-a^{P*}_{k}).\end{align} 
By using the anti-commutative relations given in (\ref{anticommutingA1}), it can be checked that $p(z_{P}(k))$ and $q(z_{P}(k))$ satisfy the Clifford relations. Since $\bar{e}^{P}_{+,k}(z^{-1})=e^{P}_{-,k}(z)$, it is not hard to see that $p(z_{P}(k))$ and $q(z_{P}(k))$ are real with respect to the conjugation, $v(z)\mapsto \bar{v}(z^{-1})$.  
It can be checked that
$(U^{P})^{2}=1$ and $$U^{P}p(z_{P}(k))(U^{P})^{-1}=-p(z_{P}(k)) \quad \mbox{and} \quad U^{P}q(z_{P}(k))(U^{P})^{-1}=-q(z_{P}(k)),$$ so $U^{P}$ is an element of the Clifford group with induced rotation $-1$ on $W^{P}$.
Furthermore, from (\ref{NPA}) we have
$$(-1)^{\mathcal{N^{P}}}F^{P}(v)(-1)^{\mathcal{N^{P}}}=-F^{P}(v)\quad\mbox{for}\quad v\in \Alt(W^{P}_{+}).$$
It follows that $F^{P}(U^{P})$ and $(-1)^{\mathcal{N}^{P}}$ have the same induced rotations in the Fock representation. Since we also have $(-1)^{2\mathcal{N}^{P}}=1$, it follows that $F^{P}(U^{P})=l_{1}(-1)^{\mathcal{N}^{P}}$ for $l_{1}=\pm 1$.  We show that $l_{1}=1$.
From (\ref{pzP}) we have that
$$a^{P*}_{k}=\tfrac{1}{2}(p(z_{P}(k))+iq(z_{P}(k)))\quad \mbox{and}\quad a^{P}_{k}=\tfrac{1}{2}(p(z_{P}(k))-iq(z_{P}(k))).$$
In the Fock representation associated with the polarization above, the number operator $\mathcal{N^{P}}$ is given by $\displaystyle{\mathcal{N^{P}}=\sum_{k=-M}^{M}a^{P*}_{k}a^{P}_{k}}$.
Using the fact that $(a^{P*}_{k}a^{P}_{k})^{2}=a^{P*}_{k}a^{P}_{k}$, following a similar argument given in \cite{palmer1}, we obtain
\begin{align*}
(-1)^{\mathcal{N^{P}}}=\prod_{k=-M}^{M}e^{i\pi a^{P*}_{k}a^{P}_{k}}=\prod_{k=-M}^{M}(1-2a^{P*}_{k}a^{P}_{k})=\prod_{k=-M}^{M}ip(z_{P}(k))q(z_{P}(k))=U^{P}.
\end{align*}
The basis elements $\{\tfrac{1}{\sqrt{2}}q(z_{A}(k)), \tfrac{1}{\sqrt{2}}p(z_{A}(k))\}$ in the $z_{A}$ coordinates are defined in an analogous way, and the second identity in (\ref{NA}) can be proved in a similar fashion.  
We now return to the proof of (\ref{detRARP3}). We start by proving that $\det R^{P}=-1$.
The linear transformation $R^{P}$ consists of a compositions of transformations, $$\xymatrix{& R^{P}:  
 \mathbb{C}^{2(2M+1)}\ar[r]^{\mathcal{F}_{P}\oplus \mathcal{F}_{P}}& \mathbb{C}^{2(2M+1)}\ar[r]^{R^{P}_{1}}& \mathbb{C}^{2(2M+1)}\ar[r]^{R^{P}_{2}}& \mathbb{C}^{2(2M+1)}\ar[r]^{x\mapsto \sqrt{2}x}& \mathbb{C}^{2(2M+1)}}.$$
Here $\mathcal{F}_{P}$ is the finite inverse Fourier transform in the periodic representation, $R^{P}_{1}$ is the transformation from the Fourier series representation to $\{e^{P}_{+,k}, e^{P}_{-,-k}\}_{k=-M}^{M}$, while $R^{P}_{2}$ is the transformation from $\{a^{P*}_{k}, a^{P}_{-k}\}_{k=-M}^{M}$ to $\{q(z_{P}(k)),p(z_{P}(k))\}_{k=-M}^{M}$.\newline
Since $\{q(z_{P}(k)),p(z_{P}(k))\}$ is not the normalized basis for $W$, but  $\bigg\{\tfrac{1}{\sqrt{2}}q(z_{P}(k)),\tfrac{1}{\sqrt{2}}p(z_{P}(k))\bigg\}$ is, the transformation $x\mapsto \sqrt{2}x$ takes the first basis to the second.  Let us denote the matrix of this transformation by $R_{3}$. (Note that the factor of $2$ is not incorporated into our definition of the Clifford relations: $xy+yx=(x,y)e$ for $x,y\in W$, where the more standard definition is $xy+yx=2(x,y)e$.) 
It follows that \begin{align}\label{R3}\det(R_{3})=2^{2M+1}.\end{align}
The inverse Fourier transform $$\mathcal{F}_{P}:\mathbb{C}^{2M+1}\to\mathbb{C}^{2M+1}$$ in the $z_{P}$ representation has the following matrix,  
\begin{align*}
\mathcal{F}_{P}=\frac{1}{\sqrt{2M+1}}\left( \begin{array}{cccccc}
z^{-M}_{P,-M}&..&..&z^{-M}_{P,M} \\
..&&&..\\
..&&&..\\
z^{M}_{P,-M}&..&..&z^{M}_{P,M} 
\end{array} \right),
\end{align*}
where $z^{l}_{P,k}:=z^{l}_{P}(k)=e^{\tfrac{2\pi ikl}{2M+1}}$ for $k,l=-M,...,M$.
It follows from \cite{MP72} that the determinant of $\mathcal{F}_{P}$ is given as
\begin{align}\label{FP}
\det (\mathcal{F}_{P})=(-i)^{M}.
\end{align}
Let $\mathcal{F}_{P}\oplus \mathcal{F}_{P}$ act on the vector $\begin{pmatrix} x_{-M} \\ y_{-M} \\ ..\\..\\x_{M}\\y_{M} \end{pmatrix}$. It can be checked that by interchanging the rows and columns an even number of times in the matrix of $\mathcal{F}_{P}\oplus \mathcal{F}_{P}$ that 
\begin{align}\label{FPFP}
\det(\mathcal{F}_{P}\oplus \mathcal{F}_{P})=
(-1)^{M}.
\end{align}
Define for $k=-M,..,M$,
$$R^{P}_{1,k}:=\frac{1}{\sqrt{2}}\left( \begin{array}{cccc}
a(z_{P}(k))&a(z_{P}(k)) \\
iz_{P}(k)a(z_{P}(k))^{-1}&-iz_{P}(k)a(z_{P}(k))^{-1}\\\end{array} \right).
$$
The matrix of $R^{P}_{1}$ is given by
$$R^{P}_{1}:=\bigg(\bigoplus_{k=-M}^{M}R^{P}_{1,k}\bigg)^{-1},$$
where\begin{align}\label{detRP1}\det R^{P}_{1}=\bigg((-i)^{2M+1}\prod_{k=-M}^{M}z_{P}(k)\bigg)^{-1}=i^{2M+1}.
\end{align}
The matrix of $R^{P}_{2}$ that goes from $\{a^{*}_{k}, a_{-k}\}_{k=-M}^{M}$ to $\{q(z_{P}(k), p(z_{P}(k)\}_{k=-M}^{M}$ is given by

$$R^{P}_{2}=\left( \begin{array}{ccccccccccc}
-i&1&0&0&..&&&&..&0&0   \\
0&0&&&&&&&&i&1\\
..&&-i&1&&&&&&&..\\
..&&&&&&&i&1&&..\\
&&&&&-i&1&&&&\\
&&&&&i&1&&&&\\
&&&&&&&-i&1&&\\
..&&i&1&&&&&&&..\\
0&0&&&&&&&&-i&1\\
i&1&..&&&&&&..&0&0\\
\end{array} \right)^{-1}.$$
\vspace{2mm}
By interchanging the rows, we obtain

\begin{align}\label{detRP2}
\det R^{P}_{2}=(-1)^{M}\det\left( \begin{array}{ccccccccccc}
-i&1&&&&   \\
i&1&&&&\\
&&..&&&\\
&&&..&\\
&&&&-i&1\\
&&&&i&1
\end{array} \right)^{-1}
 = (-1)^{M}\frac{1}{2^{2M+1}}i^{2M+1}.
\end{align}

\vspace{2mm}
It follows from (\ref{R3}), (\ref{FPFP}), (\ref{detRP1}) and (\ref{detRP2}) that
$$\det R^{P}=\det(R_{3})\det(\mathcal{F}_{P}\oplus \mathcal{F}_{P}) \det R^{P}_{1}\det R^{P}_{2}=-1.$$
The calculation of $\det R^{A}$ is similar. 
The linear transformation $R^{A}$ is a composition of transformations, $$\xymatrix{& R^{A}:  
  \mathbb{C}^{2(2M+1)}\ar[r]^{\mathcal{F}_{A}\oplus \mathcal{F}_{A}}& \mathbb{C}^{2(2M+1)}\ar[r]^{R^{A}_{1}}& \mathbb{C}^{2(2M+1)}\ar[r]^{R^{A}_{2}}& \mathbb{C}^{2(2M+1)}\ar[r]^{x\mapsto \sqrt{2}x}& \mathbb{C}^{2(2M+1)}}.$$
Here $\mathcal{F}_{A}$ is the finite inverse Fourier transform in the anti-periodic representation, $R^{A}_{1}$ is the transformation from the Fourier series representation to 
 $\{e^{A}_{+,k}, e^{A}_{-,-k}\}_{k=-M}^{M}$, while $R^{A}_{2}$ is the transformation from $\{a^{A*}_{k}, a^{A}_{-k}\}_{k=-M}^{M}$ to 
  $\{q(z_{A}(k)),p(z_{A}(k))\}_{k=-M}^{M}$. 
The matrix of $\mathcal{F}_{A}:\mathbb{C}^{2M+1}\to \mathbb{C}^{2M+1}$ is given by
\begin{align*}
\mathcal{F}_{A}:=\frac{1}{\sqrt{2M+1}}\left( \begin{array}{cccccc}
z^{-M}_{A,-M}&..&..&z^{-M}_{A,M} \\
..&&&..\\
..&&&..\\
z^{M}_{A,-M}&..&..&z^{M}_{A,M} 
\end{array} \right),
\end{align*}
where $z^{l}_{A,k}:=z^{l}_{A}(k)=e^{\tfrac{2\pi i\big(k+\tfrac{1}{2}\big)l}{2M+1}}$ for $k,l=-M,...,M$. 
We recognize this matrix as the Vandermonde matrix. It follows from \cite{BS96} that 
\begin{align}\label{VandermondeA2}
\det \mathcal{F}_{A} & \nonumber = \tfrac{1}{\sqrt{2M+1}^{2M+1}}(z^{-M}_{A,-M}....z^{-M}_{A,M})^{2M+1}\left| \begin{array}{cccccc}
1&z^{1}_{A,-M}& z^{2}_{A,-M}&....&z^{2M}_{A,-M}\\
..&&&&..\\
..&&&&..\\
1&z^{1}_{A,M}& z^{2}_{A,M}&....&z^{2M}_{A,M} 
\end{array} \right|\\ 
&  = (-1)^{M}\frac{1}{\sqrt{2M+1}^{2M+1}}\prod_{-M\leq j<k\leq M}(w^{k+\tfrac{1}{2}}-w^{j+\tfrac{1}{2}}),
\end{align}
where $w^{k}=e^{\frac{2\pi ik}{2M+1}}$.
The right hand side of the equation in (\ref{VandermondeA2}) can be written
\begin{align}\label{VandermondeA22}
(-1)^{M}\tfrac{1}{\sqrt{2M+1}^{2M+1}}(-1)^{M}\prod_{-M\leq j<k\leq M}(w^{k}-w^{j}),
\end{align} where we used that
$$\prod_{-M\leq j<k\leq M}w^{\tfrac{1}{2}}=(-1)^{\binom {2M+1} {2}}=(-1)^{M}.$$
Since $\mathcal{F}_{P}$ clearly also is a Vandermonde matrix, we have 
\begin{align}\label{VandermondeFP}\det \mathcal{F}_{P}=\tfrac{1}{\sqrt{2M+1}^{2M+1}}\prod_{-M\leq j<k\leq M}(w^{k}-w^{j}).
\end{align} It follows from (\ref{FP}), (\ref{VandermondeA2}), (\ref{VandermondeA22}) and (\ref{VandermondeFP}) that
\begin{align*}
\det \mathcal{F}_{A}=\det \mathcal{F}_{P}=(-i)^{M}.
\end{align*}
Hence, \begin{align}\label{detFFA}
\det (\mathcal{F}_{A}\oplus \mathcal{F}_{A})=(-1)^{M}.
\end{align}
Define for $k=-M,...,M$,
$$R^{A}_{1,k}:=\frac{1}{\sqrt{2}}\left( \begin{array}{cccc}
a(z_{A}(k))&a(z_{A}(k)) \\
iz_{A}(k)a(z_{A}(k))^{-1}&-iz_{A}(k)a(z_{A}(k))^{-1}\\
\end{array} \right).
$$
Then the matrix of $R^{A}_{1}$ is given by
$$R^{A}_{1}:=\bigg(\bigoplus_{k=-M}^{M}R^{A}_{1,k}\bigg)^{-1},$$
where
 \begin{align}\label{detRA1}\det R^{A}_{1}=\bigg((-i)^{2M+1}\prod_{k=-M}^{M}z_{A}(k)\bigg)^{-1}=(-i)^{2M+1}.
\end{align}
Furthermore, we have
\begin{align}\label{detRA22}
\det R^{A}_{2}=\det R^{P}_{2}=(-1)^{M}\frac{1}{2^{2M+1}}i^{2M+1}.
\end{align}
Combining  (\ref{R3}), (\ref{detFFA}), (\ref{detRA1}) and (\ref{detRA22}) we obtain
$$\det R^{A}=\det(R_{3})\det (\mathcal{F}_{A}\oplus \mathcal{F}_{A})\det R^{A}_{1}\det R^{A}_{2}=1,$$ and the theorem is proved.
\end{proof}
Let $0_{A}$ and $0_{P}$ denote the unit vacuum vectors in $\Alt(W^{A}_{+})$ and $\Alt(W^{P}_{+})$ respectively. 
We want to choose a representation of the transfer matrix $V^{A}$ in the Fock representation $\Alt(W^{A}_{+})$ such that the vacuum vector $0_{A}$ is an eigenvector for $V^{A}$ associated with its largest eigenvalue. Similarly, we want to choose a representation of the transfer matrix $V^{P}$ in the Fock representation $\Alt(W^{P}_{+})$ such that the vacuum vector $0_{P}$ is an eigenvector for $V^{P}$ corresponding to its largest eigenvalue.
By Theorem \ref{U=1}, we can write $V$ as the map,
$$V^{A}\oplus V^{P}: \Alt_{\even}(W^{A}_{+})\oplus \Alt_{\even}(W^{P}_{+}) \to \Alt_{\even}(W^{A}_{+})\oplus \Alt_{\even}(W^{P}_{+}),$$
where
$$V^{A}\simeq V|_{\Alt_{\even}(W^{A}_{+})} \quad \mbox{and} \quad V^{P}\simeq V|_{\Alt_{\even}(W^{P}_{+})}.$$
Let $T^{A}_{+}$ denote the restriction of $T^{A}$ to $W^{A}_{+}$, and let $T^{P}_{+}$ denote the restriction of $T^{P}$ to $W^{P}_{+}$.
Define as in \cite{palmer1} the linear transformations, $\Gamma(T^{A}_{+})$ and $\Gamma(T^{P}_{+})$, acting on $\Alt( W^{A}_{+})$ and $\Alt( W^{P}_{+})$ respectively
as
\begin{align}\label{TA+}
\Gamma(T^{A}_{+})=1\oplus T^{A}_{+}\oplus (T^{A}_{+}\otimes T^{A}_{+})\oplus 
\cdot \cdot \cdot \oplus  (\underbrace{T^{A}_{+}\otimes \cdot \cdot \cdot \otimes  T^{A}_{+}}_{\text{2M+1}})
\end{align} and
\begin{align}\label{TP+}
\Gamma(T^{P}_{+})=1\oplus T^{P}_{+}\oplus (T^{P}_{+}\otimes T^{P}_{+})\oplus \cdot \cdot \cdot \oplus  (\underbrace{T^{P}_{+}\otimes \cdot \cdot \cdot \otimes  T^{P}_{+}}_{\text{2M+1}}).
\end{align}
It can be checked that $$T(\Gamma(T^{A}_{+}))=T^{A}\quad \mbox{and} \quad T(\Gamma(T^{P}_{+}))=T^{P}.$$
It follows that there exists two real numbers, $\lambda^{A}_{0}$ and $\lambda^{P}_{0}$, such that the representations of $V^{A}$ and $V^{P}$ in the Fock representations $\Alt(W^{A}_{+})$ and $\Alt(W^{P}_{+})$ respectively are given by
\begin{align}\label{TATP+2}
V^{A}=\lambda^{A}_{0}\Gamma(T^{A}_{+})|_{\Alt_{\even}(W^{A}_{+})}\quad \mbox{and}\quad V^{P}=\lambda^{P}_{0}\Gamma(T^{P}_{+})|_{\Alt_{\even}(W^{P}_{+})}.
\end{align}
Since the spectra of the induced rotations $T^{A}$ and $T^{P}$  do not contain $1$, both $T^{A}_{+}$ and $T^{P}_{+}$ are strict contractions. Then it follows from (\ref{TA+}) that the largest eigenvalue of $\Gamma(T^{A}_{+})$ is $1$ with the unique eigenvector given by the vacuum vector $0_{A}$ in $\Alt(W^{A}_{+})$. Hence, $\lambda^{A}_{0}$ is the largest eigenvalue of $V^{A}$ with corresponding eigenvector $0_{A}$. A similar argument shows that $\lambda^{P}_{0}$ is the largest eigenvalue of $V^{P}$ with corresponding eigenvector $0_{P}$.
We now determine the eigenvalues for the transfer matrices $V^{A}$ and $V^{P}$.
We start by proving the following lemma.
\begin{lemma}\label{comb3} We have
\begin{align}\label{detT+AP}
\det(\Gamma(T^{P}_{+}))=(\det T^{P}_{+})^{2^{2M}}\quad \mbox{and}\quad \det(\Gamma(T^{A}_{+}))=(\det T^{A}_{+})^{2^{2M}}\quad  \mbox{for}\quad M\geq 0,
\end{align}
where $\Gamma(T^{A}_{+})$ and $\Gamma(T^{P}_{+})$ are defined in (\ref{TA+}) and (\ref{TP+}). 
\end{lemma}
\begin{proof}
We prove the first equation.
Let us consider the tensor product $(T^{P}_{+})^{\otimes k}$ which acts on $\Alt^{k}(W^{P}_{+})$. We have that $\dim (\Alt^{k}(W^{P}_{+}))=\binom{2M+1}{k}$ for $k=1,...,2M+1$. The eigenvalues of $(T^{P}_{+})^{\otimes k}$ is a product of  $e^{-\gamma(z_{P}(l_{1}))-\gamma(z_{P}(l_{2}))\cdot \cdot \cdot -\gamma(z_{P}(l_{k}))}$ with a certain multiplicity, where $-M\leq l_{1}<l_{2}<\cdot \cdot \cdot <l_{k}\leq M$. Fix a choice $e^{-\gamma(z_{P}(n))}$ in this product. We want to determine the multiplicity of $e^{-\gamma(z_{P}(n))}$. That amounts to figure out the number of ways we can combine $e^{-\gamma(z_{P}(n))}$ with the remaining factors $e^{-\gamma(z_{P}(l))}$ for $l=-M,...,M$ and $l\neq n$ in the product above. This is the same as the number of ways we can pick out $(k-1)$ elements of $e^{-\gamma(z_{P}(l))}$ from $2M$ elements, which is $\binom{2M}{k-1}$. We obtain
\begin{align*}
\det(\underbrace{T^{P}_{+}\otimes \cdot \cdot \cdot \otimes  T^{P}_{+}}_{\text{k}})=\bigg(e^{-\gamma(z_{P}(-M))}e^{-\gamma(z_{P}(-M+1))}\cdot\cdot\cdot e^{-\gamma(z_{P}(M))}\bigg)^{\binom{2M}{k-1}}=(\det(T^{P}_{+}))^{\binom{2M}{k-1}}.
\end{align*}
It follows that
\begin{align*}
\det(\Gamma(T^{P}_{+}))
& =[\det(T^{P}_{+})]^{\sum_{j=0}^{2M}\binom{2M}{j}}=(\det(T^{P}_{+}))^{2^{2M}},
\end{align*}
and hence the lemma is proved.
\end{proof}
Recall that $\Sigma_{A}$ and $\Sigma_{P}$ denote the anti-periodic and periodic spectrum respectively.
\begin{theorem}\label{spectrum}
The eigenvalues of the transfer matrices $V^{A}$ and $V^{P}$ are given by
\begin{align}\label{transfereigA}
\exp\big[\tfrac{1}{2}(\pm \gamma(z_{A}(-M))\pm ....\pm \gamma(z_{A}(M)))\big]
\end{align}
and
\begin{align}\label{transfereigP}
\exp\big[\tfrac{1}{2}(\pm \gamma(z_{P}(-M))\pm ....\pm \gamma(z_{P}(M)))\big]
\end{align}
respectively, where $\gamma$ is the positive root of
$$\cosh\gamma(z)=c^{*}_{2}c_{1}-s^{*}_{2}s_{1}\tfrac{z+z^{-1}}{2}\quad \mbox{for} \quad z \in \Sigma_{A} \quad \mbox{or}\quad z \in \Sigma_{P}.$$ 
There is an even number of minus signs in both spectra for $T<T_{C}$.
The largest eigenvalues of the transfer matrices $V^{A}$ and $V^{P}$ are given by
$$\lambda^{A}_{0}=\prod_{z\in \Sigma_{A}}e^{\tfrac{\gamma(z)}{2}} \quad \mbox{and} \quad \lambda^{P}_{0}=\prod_{z\in \Sigma_{P}}e^{\tfrac{\gamma(z)}{2}}$$ respectively. 
\end{theorem}
\begin{proof}
Recall that the transfer matrices $V^{A}$ and $V^{P}$ given in (\ref{VAVP2}) and (\ref{VAVP3}) have factors of the form $\prod e^{pq}$.
We have
$$\det e^{\sum pq}=e^{\sum \tr(pq)}.$$ Using the Clifford relations and properties of trace, we have
$$\tr(pq)=\tr(qp)=-\tr(pq)$$ which implies that $\tr(pq)=0$. Hence $\det e^{\sum pq}=1$ and it follows that 
 $\det V^{A}=\det V^{P}=1$.
 From (\ref{TATP+2}) we have the identity
 \begin{align}\label{TATP+3}
 \det V^{A}=\det[ \lambda^{A}_{0} \,\Gamma(T^{A}_{+})].
 \end{align}
 Since $\Alt(W^{A}_{+})$ has dimension $2^{2M+1}$, and $\Gamma(T^{A}_{+})$ acts on $\Alt(W^{A}_{+})$, it follows from (\ref{TATP+3}) that
 \begin{align}\label{TATP+4}
 1={(\lambda_{0}^{A})}^{{2}^{2M+1}}\det(\Gamma(T^{A}_{+})).
 \end{align}
 From Lemma \ref{comb3}, we have
 $$\det \Gamma(T^{A}_{+})=(\det T^{A}_{+})^{2^{2M}},$$ and combining this with (\ref{TATP+4}), we obtain
 $$\lambda^{A}_{0}=\frac{1}{\sqrt{\det(T^{A}_{+})}}.$$
 Since the eigenvalues of $T^{A}_{+}$ are given by the set $\{e^{-\gamma(z)}\}_{z\in \Sigma_{A}}$, we obtain
 $$\lambda^{A}_{0}=\prod_{z\in \Sigma_{A}}e^{\tfrac{\gamma(z)}{2}}.$$
It follows from (\ref{TATP+2}) that an eigenvector for $V^{A}$ is on the form, $e^{A}_{+,{j}_{1}}\wedge e^{A}_{+,{j}_{2}}\wedge...\wedge e^{A}_{+,{j}_{k}}$, for $-M\leq j_{1} <j_{2}<\cdot \cdot \cdot <j_{k}\leq M$, where $k$ is even. The corresponding eigenvalue is $\displaystyle{\exp\bigg(\frac{1}{2}\sum_{z\in \Sigma_{A}}\gamma(z)-\sum_{i=1}^{k}\gamma(z_{A}({j_{i}}))\bigg)}$.  
Hence (\ref{transfereigA}) follows. 
The calculation of the eigenvalues of $V^{P}$ is done in the exact same way.
\end{proof}
\section{Spin Operator}\label{Inducedrotationforthespinoperator}
In this section, we compute the induced rotation associated with the spin operator. We determine the matrix of this rotation in the orthonormal basis of eigenvectors for $T(V)$. 
The spin operator $\sigma_{j}$ at site $(j,0)$ acts on the Clifford generators as
\begin{align*}
\sigma_{j}p_{k}\sigma_{j}^{-1}&=-\sgn(k-j-1)p_{k},\\
\sigma_{j} q_{k}\sigma_{j}^{-1}&=-\sgn(k-j)q_{k},
\end{align*}
where
\[ \sgn(x) = \left\{ \begin{array}{ll}
         +1 & \mbox{if $x \geq 0$}\\
           -1 & \mbox{if $x < 0$}.\end{array} \right. \]
Since the spin operator $\sigma_{j}$ anti-commutes with $U$, it maps the $-1$ eigenspace for $U$ into the $+1$ eigenspace for $U.$ 
Thus, we have
$$\sigma_{j}=\left( \begin{array}{cc}
0&\sigma^{PA}_{j}  \\
\sigma^{AP}_{j} &0  \\
 \end{array} \right)$$ acting on $(U=-1)\oplus (U=+1)$. Since $(U=-1)\oplus (U=+1)$ is also unitarily equivalent to $\Alt(W^{P}_{\pm})$ and $\Alt(W^{A}_{\pm})$, we can consider $\sigma_{j}$ as the map, 
\begin{align}\label{sigma1}\sigma_{j}: \Alt(W^{P}_{+})\to \Alt(W^{A}_{+})
\end{align} or
\begin{align}\label{sigma2}\sigma_{j}: \Alt(W^{A}_{+})\to \Alt(W^{P}_{+}).
\end{align}

It is convenient to let the induced rotation for the spin operator be the identity at the endpoint $j=M$.
In order to do this, we translate the $q_{k}$ basis elements by $1$ lattice unit. This will change the operator at site $j$ to be multiplication by $-\sgn(k-j-1)$ for both $q_{k}$ and $p_{k}$. In the Fourier representation, the corresponding action on the induced rotation for the spin operator is multiplication by $z_{A}$ on the left and by $z^{-1}_{P}$ on the right.
The spin operator is then completely reduced to a ``difference'' of periodic and anti-periodic translations.
We therefore formulate the matrix representation of the induced rotation associated with $\sigma_{M}$.
Let $s:=T(\sigma_{M})$ denote multiplication by $-\sgn(k-M-1)$.
In the Fourier representation we have the following matrix representation for $s$,
\begin{align}\label{spinm1}
sf(z_{A})=\frac{1}{2M+1}\sum_{z\in\Sigma_{P}}\frac{2(z_{A}z)^{-M}}{z-z_{A}}f(z).
\end{align}
The induced rotation for the transfer matrix is then given by
$$T(\tilde{V}):=\left( \begin{array}{cc}
z&0  \\
0 &1  \\
 \end{array} \right)T(V)\left( \begin{array}{cc}
z^{-1}&0  \\
0 &1  \\
 \end{array} \right).$$
The conjugation of $T(V)$ was not introduced in Section \ref{Ising}, since it complicates the representation of the transfer matrix in the alternating tensor algebra.
The eigenvectors of $T(\tilde{V})$ corresponding to the eigenvalues $e^{-\gamma(z_{P})}$ and $e^{\gamma(z_{P})}$ are given by
\begin{align}\label{eigenP11}\tilde{e}^{P}_{+,k}(z):=\frac{1}{\sqrt{2}}\left( \begin{array}{cc}
a(z)  \\
ia(z)^{-1}   \\
 \end{array} \right)\delta_{z_{P}(k)}(z)
\end{align} and
\begin{align}\label{eigenP22}\tilde{e}^{P}_{-,k}(z):=\frac{1}{\sqrt{2}}\left( \begin{array}{cc}
a(z)  \\
-ia(z)^{-1}  \\
 \end{array} \right)\delta_{z_{P}(-k)}(z)
 \end{align} respectively for $z\in \Sigma_{P}$. The eigenvectors, $\tilde{e}^{A}_{\pm,k}(z)$, corresponding to the eigenvalues, $e^{\mp\gamma(z_{A})}$, are defined in the exact same way for $z\in \Sigma_{A}$.
We now consider the isotropic splittings,
$$W=\tilde{W}^{P}_{+}\oplus \tilde{W}^{P}_{-}\quad \mbox{and}\quad W= \tilde{W}^{A}_{+}\oplus \tilde{W}^{A}_{-},$$ where $\{\tilde{e}^{P}_{\pm,k}\}_{k=-M}^{M}$ is the choice of basis for $\tilde{W}^{P}_{\pm}$, and $\{\tilde{e}^{A}_{\pm,k}\}_{k=-M}^{M}$ is the choice of basis for $\tilde{W}^{A}_{\pm}$.
The following theorem follows.
\begin{theorem}\label{smatrix}
Suppose that $$\left( \begin{array}{cc}
A&B\\
C&D\\
\end{array} \right)$$ is the matrix of the identity map on $W$ going from the $\tilde{W}^{P}_{+}\oplus \tilde{W}^{P}_{-}$ splitting to the $\tilde{W}^{A}_{+}\oplus \tilde{W}^{A}_{-}$ splitting. Then the matrix elements of $A,B,C$ and $D$ are given by, 
\begin{align}\label{Am}
A_{lk} = D_{lk}=\frac{1}{(2M+1)}\frac{(z_{P}(k)z_{A}(l))^{-M}}{z_{P}(k)-z_{A}(l)}[a(z_{A}(l))a(z_{P}(k))^{-1}+a(z_{A}(l))^{-1}a(z_{P}(k))],
\end{align}
\begin{align}\label{Bm}
B_{lk}=\frac{1}{(2M+1)}\frac{(z_{P}(k)z_{A}(l))^{-M}}{1-z_{P}(k)z_{A}(l)}[a^{-1}(z_{A}(l))a^{-1}(z_{P}(k))-a(z_{A}(l))a(z_{P}(k))],
\end{align}
\begin{align}\label{Cm}
C_{lk}=-B_{lk}
\end{align}
for $k,l=-M,....,M$
with $$a(z)=\sqrt{\frac{\mathcal{A}_{1}(z)\mathcal{A}_{2}(z)}{\mathcal{A}_{1}(z^{-1})\mathcal{A}_{2}(z^{-1})}}\quad \mbox{and} \quad \mathcal{A}_{j}=\sqrt{\alpha_{j}-z},\quad\quad j=1,2$$  for $z \in \Sigma_{A}$ or $z\in \Sigma_{P}$.
\end{theorem}
Here we have for example $D=Q^{A}_{-}|_{\tilde{W}^{P}_{-}}$ which maps $\tilde{W}^{P}_{-}$ into $\tilde{W}^{A}_{-}$.
Recall that the representations $\{q_{k}, p_{k}\}$ are generators of an irreducible $*$-representation of the Clifford algebra $\Cliff(W)$ on $\mathbb{C}^{2(2M+1)}$ \cite{palmer1}.
Let $F$ denote the Fock representation associated with these Clifford relations. Since
$\sigma_{M}x\sigma^{-1}_{M}=sx$ for $x\in W$, where $\sigma_{M}:\mathbb{C}^{2(2M+1)}\to \mathbb{C}^{2(2M+1)}$, it follows that
\begin{align}\label{Fock11}
\sigma_{M}F(x)=F(s(x))\sigma_{M}.
\end{align}
The Fock representation defined in (\ref{Fock7}) is an irreducible $*$-representation of the Clifford algebra. All irreducible $*$ representations of the Clifford algebra over $W$ are unitarily equivalent \cite{BW35}. Thus, it follows that there exists a unitary map, $$U_{P}:\mathbb{C}^{2(2M+1)}\to \Alt(\tilde{W}^{P}_{+})$$ such that
\begin{align}\label{Fock22}
F^{P}(x)=U_{P}F(x)U^{-1}_{P}
\end{align} and a unitary map, $$U_{A}:\mathbb{C}^{2(2M+1)}\to \Alt(\tilde{W}^{A}_{+})$$ such that
\begin{align}\label{Fock33}
F^{A}(s(x))=U_{A}F(s(x))U^{-1}_{A}.
\end{align}
              It follows from (\ref{Fock11}), (\ref{Fock22}) and (\ref{Fock33}) that
              $$U_{A}\sigma_{M}U^{-1}_{P}F^{P}(x)=F^{A}(s(x))U_{A}\sigma_{M}U^{-1}_{P}.$$
              The map,
              $$\sigma:=U_{A}\sigma_{M}U^{-1}_{P}:\Alt(\tilde{W}^{P}_{+})\to \Alt(\tilde{W}^{A}_{+}),$$ is an intertwining map which is unique up to a constant by Schur's lemma.

Thus, we have the following commutative diagram,
$$\xymatrix{  
 &  \Alt (\tilde{W}^{P}_{+} )\ar[r]^{\sigma} & \Alt (\tilde{W}^{A}_{+}) \\
              & \Alt ( \tilde{W}^{P}_{+})  \ar[r]^{\sigma} \ar[u]_{F^{P}} & \Alt (\tilde{W}^{A}_{+})  
\ar[u]_{F^{A}\circ s},\\
              }$$
             where $\sigma$ satisfies the intertwining relation,
              \begin{align}\label{Fockid}
\sigma F^{P}(x)=F^{A}(s(x))\sigma.
\end{align}
We will use this relation in Section \ref{Spinmatrixelements} when we find an expression for the spin matrix elements.
Now we can regard $\sigma^{AP}_{j}$ as the restriction of $\sigma_{j}$ to $\Alt_{\even}(\tilde{W}^{P}_{+})$. In a similar fashion, we can restrict $\sigma_{j}$ to $\Alt_{\even}(\tilde{W}^{A}_{+})$ and recover $\sigma^{PA}_{j}$.
This restriction is central since we will consider the spin operator in terms of an eigenvector basis for the transfer matrix $V$ which is simple only in $\Alt_{\even}(\tilde{W}^{A}_{+})$ and $\Alt_{\even}(\tilde{W}^{P}_{+})$.
We are interested in understanding the $N\to \infty$ limit of the two point correlation for the spin operator. The eigenvector associated with the largest eigenvalue for the transfer matrix will dominate in this limit under the action of the spin operator $\sigma$. In order to compute this limit, we will need to know the spin matrix elements for $\sigma$. As we will see in Section \ref{Spinmatrixelements}, the spin matrix elements can be written in terms of the inverse of $D$ in (\ref{Am}).
\section{Two-point Correlation Function}\label{Twopoint}
For $n>m$, we have \cite{kaufman}
\begin{align}\label{paircorr}
\langle \sigma_{mi}\sigma_{nj}\rangle_{\Lambda}=\frac{\tr(\sigma_{i}V^{n-m}\sigma_{j}V^{(2N+1)-n+m})}{\tr(V^{2N+1})}.
\end{align}
Recall that
$$V=V^{A}\oplus V^{P},$$ where $$V^{A}=\lambda^{A}_{0}\Gamma(T^{A}_{+})|_{\Alt_{\even}(W^{A}_{+})},\quad V^{P}=\lambda^{P}_{0}\Gamma(T^{P}_{+})|_{\Alt_{\even}(W^{P}_{+})},$$ and where $\Gamma(T^{A}_{+})$ and $\Gamma(T^{P}_{+})$ are defined in (\ref{TA+}) and (\ref{TP+}). Using the equation above, the identity $\Gamma(T^{A}_{+})0_{A}=0_{A},$ and the fact that only the terms which involve the vacuum vector $0_{A}$ survive in the semi-infinite volume limit, $N\to\infty$, we obtain
\begin{align*}
\lim_{N\to \infty}\langle \sigma_{mi}\sigma_{nj}\rangle_{\Lambda}& \nonumber =\langle 0_{A}, \sigma_{i}\big(\tfrac{V}{\lambda^{A}_{0}}\big)^{n-m}\sigma_{j}0_{A} \rangle\\
&= e^{\tfrac{1}{2}(n-m)\sum_{j=-M}^{M}\big[\gamma(z_{P}(j))-\gamma(z_{A}(j))\big]}\times\\
&\nonumber \times\sum_{K=1}^{M}\sumtwo{-M\leq k_{1}<\cdot\cdot\cdot}{\cdot\cdot\cdot<k_{2K}\leq M}\langle 0_{A},\sigma_{i}e^{P}_{+,{k}_{1}}\wedge\cdot\cdot\cdot \wedge e^{P}_{+,{k}_{2K}}\rangle \times\\
& \times e^{-(n-m)\sum_{j=1}^{2K}\gamma(z_{P}(k_{j}))}\langle e^{P}_{+,{k}_{1}}\wedge\cdot\cdot\cdot \wedge e^{P}_{+,{k}_{2K}}\sigma_{j},0_{A}\rangle.
\end{align*}
On the torus, where both $M$ and $N$ are fixed, the two-point correlation function can be written in a form involving the spin matrix elements
$\langle e^{A}_{+,{l}_{1}}\wedge\cdot\cdot\cdot\wedge e^{A}_{+,{l}_{L}},\sigma e^{P}_{+,{k}_{1}}\wedge\cdot\cdot\cdot\wedge e^{P}_{+,{k}_{K}}\rangle$
for $L+K=\even$.

\section{Spin Matrix Elements on the Finite Periodic Lattice}\label{Spinmatrixelements}
In this section, we will provide a formula for the matrix elements of the spin operator in a basis of eigenvectors for the transfer matrix $V$.
Recall that after ``conjugation'' the spin operator is a map,
$$\sigma: \Alt(\tilde{W}^{P}_{+}) \to \Alt(\tilde{W}^{A}_{+}),$$ and the induced rotation, $s:=T(\sigma)$, associated with the spin operator is the identity on $W$. We defined the matrix of this map to be
$$s:=\left( \begin{array}{cc}
A & B  \\
C & D  \\
\end{array} \right):\tilde{W}^{P}_{+}\oplus \tilde{W}^{P}_{-}\to \tilde{W}^{A}_{+}\oplus \tilde{W}^{A}_{-},$$ where for example $D=Q^{A}_{-}|_{\tilde{W}^{P}_{-}}$. Let $D^{\tau}$ denote the transpose of $D$ with respect to the bilinear form $(\cdot,\cdot)$, and $D^{-\tau}:\tilde{W}^{P}_{+}\to \tilde{W}^{A}_{+}$ is the inverse of $D^{\tau}$. Here $D$ is invertible when the one-point function $\langle 0_{A}, \sigma 0_{P} \rangle$ is nonzero. Let $\mathcal{P}$ denote the collection of subsets of \newline
$\{ -M,-M+1,.... ,M\}$.
For an element $J$ in $\mathcal{P}$, we write
$$J=\{J_{1}, J_{2},...,J_{k}\}\quad \mbox{with}\quad J_{1}<J_{2}< \cdot \cdot \cdot<J_{k}.$$
We write $\#J=k$ for the number of elements in $J$. If $R$ is a $2M\times 2M$ matrix, we let $R_{I,J}$ denote the $(\#I+\#J) \times (\#I+\#J)$ sub matrix of $R$ made from the rows and columns of $R$ indexed by $I$ and $J$ respectively. For $I,J\in \mathcal{P}$
define
$$\tilde{e}^{A}_{\pm,I}:=\tilde{e}^{A}_{{\pm,I}_{1}}\wedge \tilde{e}^{A}_{{\pm,I}_{2}}\wedge \cdot \cdot \cdot\wedge \tilde{e}^{A}_{{\pm,I}_{k}}$$ and
$$\tilde{e}^{P}_{\pm,J}:=\tilde{e}^{P}_{{\pm,J}_{1}}\wedge \tilde{e}^{P}_{{\pm,J}_{2}}\wedge \cdot \cdot \cdot \wedge \tilde{e}^{P}_{{\pm,J}_{k}}.$$ 
The sets $\{\tilde{e}^{A}_{\pm,I}\}$ and $\{\tilde{e}^{P}_{\pm,J}\}$ are orthonormal bases for $\Alt(\tilde{W}^{A}_{\pm})$ and $\Alt(\tilde{W}^{P}_{\pm})$ respectively.
We have the following theorem.
\begin{theorem}\label{spinmatrix5}
Let $I,J\in \mathcal{P}$ and suppose  $\langle 0_{A}, \sigma 0_{P}\rangle \neq 0$.
Then the spin matrix elements are given by
$$\langle\tilde{e}^{A}_{+,I},\sigma\tilde{e}^{P}_{+,J}\rangle =\langle 0_{A}, \sigma 0_{P} \rangle\Pf(R_{I,J}),$$
where $\Pf(R_{I,J})$ is the Pfaffian of the skew symmetric matrix,
$$R_{I,J}=\left( \begin{array}{cc}
BD^{-1}_{I\times I} & D^{-\tau}_{I\times J}  \\
-D^{-1}_{J\times I} & D^{-1}C_{J\times J}  \\
\end{array} \right).$$ 
The sum is over all such $I$ and $J$ with $\#I+\#J$ even.
\end{theorem}
\begin{proof}
Since $\sigma$ satisfies the intertwining relation  (\ref{Fockid}), the proof follows from Theorem \ref{gPfaffian}.
\end{proof}
\begin{corollary}\label{Cor}
For $-M \leq k< j\leq M$ and $\langle 0_{A}, \sigma 0_{P}\rangle\neq 0$, we have
\begin{align*}
\frac{\langle 0_{A}, \sigma \,\tilde{e}^{P}_{+,k} \tilde{e}^{P}_{+,j}  \rangle}{\langle 0_{A}, \sigma \, 0_{P}\rangle}=( \,\tilde{e}^{P}_{+,k}, D^{-1}C\, \tilde{e}^{P}_{+,j}  )=(D^{-1}C)_{kj},
\end{align*}
\begin{align*}
\frac{\langle \tilde{e}^{A}_{+,k} \tilde{e}^{A}_{+,j}, \sigma \,0_{P} \rangle}{\langle 0_{A}, \sigma \,0_{P}\rangle}=( \,\tilde{e}^{A}_{-,k}, BD^{-1}\, \tilde{e}^{A}_{-,j}  )=(BD^{-1})_{kj},
\end{align*}
\begin{align}\label{Dm}
\frac{\langle \tilde{e}^{A}_{+,k},\sigma\, \tilde{e}^{P}_{+,j} \rangle}{\langle 0_{A}, \sigma \, 0_{P}\rangle}=( \,\tilde{e}^{A}_{-,k}, D^{-\tau}\, \tilde{e}^{P}_{+,j}  )=D^{-\tau}_{kj}.
\end{align} 
\end{corollary}
\section{Bugrij-Lisovyy Conjecture for the Spin Matrix Elements}\label{BLformula}
For simplicity, let $\Sigma_{P}$ and $\Sigma_{A}$ now denote the sets $\bigg\{\tfrac{2 \pi k}{2M+1}, k \in \{-M,...,M\}\bigg\}$ and \newline
$\bigg\{\tfrac{2 \pi \big(k+\tfrac{1}{2}\big)}{2M+1}, k \in \{-M,...,M\}\bigg\}$ respectively. 
We consider the isotropic case and define the interaction constant to be 
$\mathcal{K}:=\mathcal{K}_{1}=\mathcal{K}_{2}.$
The function $\gamma(\theta)$ is defined as the positive root of the equation
$$\cosh(\gamma(\theta))=\sinh(2\mathcal{K})+\sinh(2\mathcal{K})^{-1}-\cos(\theta).$$
A. I. Bugrij and O. Lisovyy \cite{BL03} proposed the following formula for the spin matrix elements on the finite periodic lattice for the isotropic case in the orthonormal basis of eigenvectors for the transfer matrix,
\begin{align}\label{BL1}
&\nonumber\mathrel{\phantom{=}}\langle \tilde{e}^{A}_{+,{l}_{1}}\wedge \tilde{e}^{A}_{+,{l}_{2}}\wedge...\wedge \tilde{e}^{A}_{+,{l}_{m}},\sigma \tilde{e}^{P}_{+,{k}_{1}}\wedge \tilde{e}^{P}_{+,{k}_{2}}\wedge...\wedge \tilde{e}^{P}_{+,{k}_{m'}} \rangle\\
& \nonumber=\sqrt{\xi\xi_{T}}\prod_{i=1}^{m}\frac{e^{\tfrac{1}{2}v\big(e^{i\theta^{A}_{{l_{i}}}}\big)}}{\sqrt{(2M+1) \sinh\gamma(\theta^{A}_{{l}_{i}})}}\prod_{j=1}^{m'}\frac{e^{-\tfrac{1}{2}v\big(e^{i\theta^{P}_{{k_{j}}}}\big)}}{\sqrt{(2M+1) \sinh\gamma(\theta^{P}_{{k}_{j}})}}\times\\
& \nonumber\times\prod_{1\leq i<i'\leq m}\frac{\sin\tfrac{\theta^{A}_{{l}_{i}}-\theta^{A}_{{l}_{i'}}}{2}}{\sinh\tfrac{\gamma(\theta^{A}_{{l}_{i}})+\gamma(\theta^{A}_{{l}_{i'}})}{2}} \prod_{1\leq j<j'\leq m'}\frac{\sin\tfrac{\theta^{P}_{{k}_{j}}-\theta^{P}_{{k}_{j'}}}{2}}{\sinh\tfrac{\gamma(\theta^{P}_{{k}_{j}})+\gamma(\theta^{P}_{{k}_{j'}})}{2}}\times\\
&\times\prod_{1\leq i\leq m, 1\leq j\leq m'}\frac{\sinh\tfrac{\gamma(\theta^{A}_{{l}_{i}})+\gamma(\theta^{P}_{{k}_{j}})}{2}}{\sin\tfrac{\theta^{A}_{{l}_{i}}-\theta^{P}_{k_{j}}}{2}},
\end{align}
where
$$\xi=|1-(\sinh(2\mathcal{K}))^{-4}|^{\tfrac{1}{4}},$$ and where
the cylindrical parameters $\xi_{T}$ and $v(e^{i\theta})$ are defined as
\begin{align*}
\xi^{4}_{T}:=\frac{\prod_{\theta'\in \Sigma_{P}}\prod_{\theta\in \Sigma_{A}}\sinh^{2}\big(\tfrac{\gamma(\theta')+\gamma(\theta)}{2}\big)}{\prod_{\theta'\in \Sigma_{P}}\prod_{\theta\in \Sigma_{P}}\sinh\big(\tfrac{\gamma(\theta')+\gamma(\theta)}{2}\big)\prod_{\theta\in \Sigma_{A}}\prod_{\theta'\in \Sigma_{A}}\sinh\big(\tfrac{\gamma(\theta)+\gamma(\theta')}{2}\big)},
\end{align*} 
\begin{align*}
v(e^{i\theta})=\log\frac{\prod_{\theta' \in \Sigma_{A}}\sinh\big(\tfrac{\gamma(\theta)+\gamma(\theta')}{2}\big)}{\prod_{\theta' \in \Sigma_{P}}\sinh \big(\tfrac{\gamma(\theta)+\gamma(\theta')}{2}\big)}.
\end{align*}
Here $m+m'$ is even.
It can be checked that $e^{\pm\tfrac{1}{2}v(e^{i\theta})}$ can be written in the forms,
\begin{align}\label{V+}
V_{+}(e^{i\theta}):=e^{\tfrac{v(e^{i\theta})}{2}}=\sqrt{\frac{e^{-\tfrac{\gamma(\pi)}{2}}(\lambda^{-1}(e^{i\theta})-\lambda(-1))}{e^{-\tfrac{\gamma(0)}{2}}(\lambda^{-1}(e^{i\theta})-\lambda(1))}}\frac{\prod_{\theta' >0\in \Sigma_{A}}e^{-\tfrac{\gamma(\theta')}{2}}(\lambda^{-1}(e^{i\theta})-\lambda(e^{i\theta'}))}{\prod_{\theta'>0 \in \Sigma_{P}}e^{-\tfrac{\gamma(\theta')}{2}}(\lambda^{-1}(e^{i\theta})-\lambda(e^{i\theta'}))}
\end{align}and
\begin{align}\label{V-}
V_{-}(e^{i\theta}):=e^{-\tfrac{v(e^{i\theta})}{2}}=\sqrt{\frac{e^{\tfrac{\gamma(0)}{2}}(\lambda(e^{i\theta})-\lambda^{-1}(1))}{e^{\tfrac{\gamma(\pi)}{2}}(\lambda(e^{i\theta})-\lambda^{-1}(-1))}}\frac{\prod_{\theta' >0\in \Sigma_{P}}e^{\tfrac{\gamma(\theta')}{2}}(\lambda(e^{i\theta})-\lambda^{-1}(e^{i\theta'}))}{\prod_{\theta' >0\in \Sigma_{A}}e^{\tfrac{\gamma(\theta')}{2}}(\lambda(e^{i\theta})-\lambda^{-1}(e^{i\theta'}))},
\end{align}
where $\lambda(e^{i\theta})=e^{\gamma(\theta)}$ for $\theta=\theta^{A}_{{l}_{i}}$ or $\theta=\theta^{P}_{{k}_{j}}$, and where we used the short-hand notation $\gamma(\theta):=\gamma(e^{i\theta})$.
Here the square roots are taken with positive real parts.
This particular form of $v(e^{i\theta})$ will play a central role in Section \ref{Spectral2} when we consider the factorization of the ratio of summability kernels on the spectral curve.
\section{Pfaffian Formalism and the Bugrij-Lisovyy Conjecture}\label{Pfaffianformalism}
The Bugrij-Lisovyy formula in (\ref{BL1}) and Corollary \ref{Cor} lead to the following conjecture for the isotropic Ising model:
\begin{conjecture}
The matrix elements of $D^{-\tau}$, $BD^{-1}$ and $D^{-1}C$ are given by,
\begin{align}\label{11}
&\mathrel{\phantom{=}}D^{-\tau}_{{i},{j}}
\propto \frac{V_{+}(z_{A}(i))}{\sqrt{(2M+1) \sinh\gamma(\theta^{A}_{i})}}\frac{V_{-}(z_{P}(j))}{\sqrt{(2M+1) \sinh\gamma(\theta^{P}_{j})}}\frac{\sinh\tfrac{\gamma(\theta^{A}_{i})+\gamma(\theta^{P}_{j})}{2}}{\sin\tfrac{\theta^{A}_{i}-\theta^{P}_{j}}{2}},
\end{align}
\begin{align*}
&\nonumber\mathrel{\phantom{=}}BD^{-1}_{{i},{i'}}\\
& \propto \frac{V_{+}(z_{A}(i))}{\sqrt{(2M+1) \sinh\gamma(\theta^{A}_{i})}}\frac{V_{+}(z_{A}(i'))}{\sqrt{(2M+1) \sinh\gamma(\theta^{A}_{i'})}}
\frac{\sin\tfrac{\theta^{A}_{{i}}-\theta^{A}_{i'}}{2}}{\sinh\tfrac{\gamma(\theta^{A}_{i})+\gamma(\theta^{A}_{{i'}})}{2}},
\end{align*}
\begin{align*}
&\nonumber \mathrel{\phantom{=}}D^{-1}C_{{j},{j'}}\\
& \propto \frac{V_{-}(z_{P}(j))}{\sqrt{(2M+1) \sinh\gamma(\theta^{P}_{j})}}\frac{V_{-}(z_{P}(j'))}{\sqrt{(2M+1) \sinh\gamma(\theta^{P}_{j'})}}
\frac{\sin\tfrac{\theta^{P}_{j}-\theta^{P}_{j'}}{2}}{\sinh\tfrac{\gamma(\theta^{P}_{j})+\gamma(\theta^{P}_{j'})}{2}}
\end{align*} respectively,
where $\propto$ denotes proportional to.
\end{conjecture} The proportionality constants in each of the terms above are the same and therefore omitted.
Here we have extended the Bugrij-Lisovyy conjecture to include the elements 
$\langle \tilde{e}^{A}_{+,k},\sigma\, \tilde{e}^{P}_{+,j} \rangle$ which are not ``physical elements''.
By applying Theorem \ref{spinmatrix5} and Wicks Theorem, we can express the spin matrix elements
 $\langle \tilde{e}^{A}_{+,{l}_{1}}\wedge...\wedge \tilde{e}^{A}_{+,{l}_{m}},\sigma \tilde{e}^{P}_{+,{k}_{1}}\wedge ...\wedge \tilde{e}^{P}_{+,{k}_{m'}} \rangle$ in terms of Pfaffians of a skew symmetric matrix whose elements are multiples of $\langle \tilde{e}^{A}_{+,{l}_{i}}\wedge \tilde{e}^{A}_{+,{l}_{i'}}\sigma 0_{P}\rangle$,
 $\langle 0_{A}, \sigma \tilde{e}^{P}_{+,{k}_{j}}\wedge \tilde{e}^{P}_{+,{k}_{j'}}\rangle$ and $\langle \tilde{e}^{A}_{+,{l}_{i}},\sigma \tilde{e}^{P}_{+,{k}_{j}}\rangle$. By using the proposed Bugrij-Lisovyy formula for these elements, we show that the Pfaffian of the matrix can be written as a product of Jacobian elliptic functions in the uniformization parameter of the spectral curve. We find that up to a constant, this product is the same as the one given in the Bugrij-Lisovyy formula. This supports the Bugrij-Lisovyy conjecture. We refer the reader to Appendix \ref{Appendix C} for a description of elliptic functions and the Uniformization Theorem \ref{palmer}. We start by writing the factor
$$\frac{\sinh(\tfrac{\gamma(\theta)+\gamma(\theta')}{2})}{\sin(\tfrac{\theta-\theta'}{2})}$$ in (\ref{11}) in terms of the uniformization parameters by performing the elliptic substitutions,
$$e^{i\theta }=z(u,a)=k\sn(u+ia)\sn(u-ia),\quad e^{i\theta'}=z(u',a)=k\sn(u'+ia)\sn(u'-ia),$$
where
$$\Im u=\Im u'=\frac{K'}{2}\quad \mbox{and}\quad \Re u \in (0,2K),$$ and $0< 2a< K'$ is defined by $s_{1}=-i\sn(2ia)$. Here $K$ and $K'$ are elliptic integrals given in (\ref{K}) and (\ref{K'}), and $\sn(u)$ is a Jacobian elliptic function of $u$ with modulus $k=\frac{1}{s_{1}s_{2}}$.
Then Theorem \ref{palmer} implies
\begin{align*}
e^{\tfrac{\gamma+\theta i}{2}} & = \sqrt{k}\sn(u-ia),\\
e^{\tfrac{-\gamma+\theta i}{2}} & = \sqrt{k}\sn(u+ia).
\end{align*}
We have
\begin{align}\label{hyperbolic}
\frac{\sinh(\tfrac{\gamma(\theta)+\gamma(\theta')}{2})}{\sin(\tfrac{\theta-\theta'}{2})}=i\frac{e^{\tfrac{\gamma+\gamma'}{2}}-e^{-\tfrac{\gamma+\gamma'}{2}}}{e^{\tfrac{i\theta -i\theta'}{2}}-e^{-\tfrac{i\theta -i\theta' }{2}}}& = i\frac{e^{\tfrac{\gamma+i\theta }{2}}e^{\tfrac{\gamma'+i\theta' }{2}}-e^{-\tfrac{\gamma+i\theta }{2}}e^{-\tfrac{\gamma'+i\theta'}{2}}}{e^{i\theta }-e^{i\theta' }}.
\end{align}
Thus the numerator of (\ref{hyperbolic}) becomes 
$$ik[\sn(u-ia)\sn(u'-ia)-\sn(u'+ia)\sn(u+ia)].$$
Using the identity
$$\sn(u-ia)\sn(u+ia)=\frac{\sn^{2}(u)-\sn^{2}(ia)}{1-k^{2}\sn^{2}(u)\sn^{2}(ia)}$$ which follows from the addition formulas (\ref{additionsn}) and (\ref{subtractionsn}),
the denominator becomes
$$k\bigg(\frac{\sn^{2}(u)-\sn^{2}(ia)}{1-k^{2}\sn^{2}(u)\sn^{2}(ia)}\bigg)-k\bigg(\frac{\sn^{2}(u')-\sn^{2}(ia)}{1-k^{2}\sn^{2}(u')\sn^{2}(ia)}\bigg).$$
After some simplifications, the expression in (\ref{hyperbolic}) becomes
\begin{align}\label{snin}
&\nonumber\mathrel{\phantom{=}}\frac{-2i[\cn(ia)\dn(ia)\sn(ia)][\sn(u)\cn(u')\dn(u')+\sn(u')\cn(u)\dn(u)]}{(\sn^{2}(u)-\sn^{2}(u'))(1-k^{2}\sn^{4}(ia))}\\
& = s_{1}\frac{1}{\sn(u-u')}.
\end{align}
In the isotropic case we have $s_{1}^{-1}=\sqrt{k}$.
If we substitute $v:=-u+\frac{iK'}{2}$ and 
$v':=-u'-\frac{iK'}{2}$ in (\ref{snin}), and apply identity (\ref{translation}), we obtain
$$s_{1}\frac{1}{\sn(u-u')}=-\sqrt{k}\sn(v-v').$$
Now we will make the following elliptic substitutions,
\begin{align*}
e^{i\theta^{A}_{{l}_{j}}} & = k\sn(u_{{l}_{j}}-ia)\sn(u_{{l}_{j}}+ia)\quad \mbox{for}\quad 1\leq j\leq m,\\
e^{i\theta^{P}_{{k}_{j}}}& = k\sn(u_{{k}_{j}}-ia)\sn(u_{{k}_{j}}+ia)\quad \mbox{for}\quad 1\leq j\leq m',
\end{align*} followed by the translations, 
\begin{align*}
v_{{l}_{j}}: & = -u_{{l}_{j}}+\frac{iK'}{2}\quad \mbox{for}\quad 1\leq j\leq m,\\
v_{{k}_{j}}: & =- u_{{k}_{j}}-\frac{iK'}{2}\quad \mbox{for} \quad 1\leq j\leq m'.
\end{align*}
Then using Theorem \ref{spinmatrix5} and Wicks Theorem, we obtain for $m+m'=\mbox{even}$,
\begin{align*}
&\mathrel{\phantom{=}}\langle \tilde{e}^{A}_{+,{l}_{1}}\wedge \tilde{e}^{A}_{+,{l}_{2}}\wedge...\wedge \tilde{e}^{A}_{+,{l}_{m}},\sigma \tilde{e}^{P}_{+,{k}_{1}}\wedge \tilde{e}^{P}_{+,{k}_{2}}\wedge...\wedge \tilde{e}^{P}_{+,{k}_{m'}} \rangle\\
&=\langle 0_{A}, \sigma 0_{P}\rangle\Pf(r),
\end{align*} where $\Pf(r)$ is the Pfaffian of the skew symmetric matrix $r$ with matrix elements above the diagonal given by:
For $1\leq i<i'\leq m$, \begin{align*}r_{{l}_{i}, {l}_{i'}}&=\frac{\langle \tilde{e}^{A}_{+,{l}_{i}}  \tilde{e}^{A}_{+,{l}_{i'}},\sigma 0_{P}\rangle}{\langle 0_{A},\sigma 0_{P}\rangle}\\
& \propto \frac{V_{+}(z_{A}({l_{i})})}{\sqrt{(2M+1) \sinh\gamma(\theta^{A}_{{{l}_{i}}}})}\frac{V_{+}(z_{A}(l_{i'}))}{\sqrt{(2M+1) \sinh\gamma(\theta^{A}_{{{l}_{i'}}})}}(-\sqrt{k}\sn(v_{{l}_{i}}-v_{{l}_{i'}})),
\end{align*}
for $1\leq i\leq m\quad \mbox{and} \quad 1\leq j\leq m'$,
\begin{align*}
r_{{l}_{i},{m+k}_{j}}&=\frac{\langle \tilde{e}^{A}_{+,{l}_{i}},\sigma \tilde{e}^{P}_{+,{k}_{j}}\rangle}{\langle 0_{A},\sigma 0_{P}\rangle}\\
& \propto \frac{V_{+}(z_{A}(l_{i}))}{\sqrt{(2M+1) \sinh\gamma(\theta^{A}_{{{l}_{i}}})}}\frac{V_{-}(z_{P}(k_{j}))}{\sqrt{(2M+1) \sinh\gamma(\theta^{P}_{{{k}_{j}}})}}(-\sqrt{k}\sn(v_{{l}_{i}}-v_{{k}_{j}}))
\end{align*}
and
for $1 \leq j< j' \leq m'$,
\begin{align*}
r_{{m+k}_{j},{m+k}_{j'}}&=\frac{\langle 0_{A},\sigma \tilde{e}^{P}_{+,{k}_{j}} \tilde{e}^{P}_{+,{k}_{j'}}\rangle}{\langle 0_{A},\sigma 0_{P}\rangle}\\
& \propto  \frac{V_{-}(z_{P}(k_{j}))}{\sqrt{(2M+1) \sinh\gamma(\theta^{P}_{{{k}_{j}}})}}\frac{V_{-}(z_{P}(k_{j'}))}{\sqrt{(2M+1) \sinh\gamma(\theta^{P}_{{{k}_{j'}}})}}(-\sqrt{k}\sn(v_{{k}_{j}}-v_{{k}_{j'}})).
\end{align*}
Define $E$ to be the $(m+m')\times (m+m')$ diagonal matrix
$$E=\left( \begin{array}{ccc}
e^{A} & 0 \\
0 & e^{P}  \end{array} \right)$$ with matrix elements
$$e^{A}_{{l}_{i},{l}_{i}}=\frac{V_{+}(z_{A}(l_{i}))}{\sqrt{(2M+1)\sinh \gamma(\theta^{A}_{{l}_{i}})}}\quad \mbox{for}\quad 1\leq i\leq m,$$
$$e^{P}_{m+{k}_{j},{m+k}_{j}}=\frac{V_{-}(z_{P}(k_{j}))}{\sqrt{(2M+1)\sinh\gamma(\theta^{P}_{{k}_{j}})}}\quad\mbox{for}\quad 1\leq j\leq m'.$$
Define $s$ to be the $(m+m')\times (m+m')$ skew symmetric matrix with matrix elements above the diagonal given by 
\begin{align*}
s_{{l}_{i},{l}_{i'}} & = -\sqrt{k}\sn(v_{{l}_{i}}-v_{{l}_{i'}})\quad \mbox{for}\quad 1\leq i<i' \leq m,\\
s_{{l}_{i,}{m+k}_{j}} & = -\sqrt{k}\sn(v_{{l}_{i}}-v_{{k}_{j}})\quad \mbox{for}\quad 1\leq i \leq m, \quad 1\leq j\leq m',\\
s_{m+{k}_{j},{m+k}_{j'}} & = -\sqrt{k}\sn(v_{{k}_{j}}-v_{{k}_{j'}})\quad \mbox{for}\quad 1\leq j<j' \leq m'.
\end{align*}
We will need the following lemma found on page $87$ of \cite{palmer1}.
\begin{lemma}\cite{palmer1}\label{Li}
Let $r$ be the  $2n \times 2n$ skew symmetric matrix with $i,j$ matrix element $r_{i,j}=-\sqrt{k}\sn(u_{i}-u_{j})$ for $i,j=1,2,...,2n.$
Then
$$\Pf(r)=\prod_{i<j}^{2n}r_{i,j}.$$
\end{lemma}
Using Lemma \ref{Li}, we have
\begin{align*}\Pf(r)&=\Pf(EsE^{\tau})\\
&=(\det E)\Pf(s)\\
& = \prod_{i=1}^{m}\frac{V_{+}(z_{A}(l_{i}))}{\sqrt{(2M+1) \sinh\gamma(\theta^{A}_{{l}_{i}})}}\prod_{j=1}^{m'}\frac{V_{-}(z_{P}(k_{j}))}{\sqrt{(2M+1) \sinh\gamma(\theta^{P}_{{k}_{j}})}}\times\\
& \times \prod_{1\leq i < i'\leq m}-\sqrt{k}\sn(v_{{l}_{i}}-v_{{l}_{i'}})\prod_{1\leq j < j'\leq m'}-\sqrt{k}\sn(v_{{k}_{j}}-v_{{k}_{j'}})\times\\
& \times \prod_{1\leq i \leq m, 1\leq j \leq m'}-\sqrt{k}\sn(v_{{l}_{i}}-v_{{k}_{j}}),
\end{align*}
which up to a constant is the same product formula as given in (\ref{BL1}).
\section{Numerical Calculations}\label{Numerical}
We have numerically compared the calculation of our $BD^{-1}$, $D^{-\tau}$, and $D^{-1}C$ matrix elements with the corresponding terms in the Bugrij-Lisovyy formula  (\ref{BL1}). 
What we find is, that as we let the temperature approach $T_{C}$ from below, our results are close to theirs,  but there is a discrepancy. We are uncertain of the reason for this. We also find that in the scaling regime, i.e. as we let $M$ get bigger while we keep the temperature close to $T_{C}$, the precision improves. We compared the formulas in MATLAB by creating our $(2M+1)\times (2M+1)$ matrix $D$ as given in (\ref{Am}). Then we multiplied the transpose of this matrix with the matrix with elements $\langle \tilde{e}^{A}_{+,k}, \sigma \tilde{e}^{P}_{+,j} \rangle$ for $-M\leq k,j \leq M$ as given in the Bugrij-Lisovyy formula (\ref{BL1}). If this had returned the identity matrix, it would have shown  numerically that the Bugrij-Lisovyy formula provides an inverse for our matrix $D^{\tau}$.

\section{Holomorphic Factorization of the Ratio of Summability Kernels on the Spectral Curve}\label{Spectral2}
In this section, we exhibit the ``new'' elements $V_{\pm}$ in the Bugrij-Lisovyy formula as part of a holomorphic factorization of the periodic and anti-periodic summability kernels on the spectral curve associated with the induced rotation for the transfer matrix.

Introduce the cycles
\begin{align*}
\mathcal{N}_{+}&=\big\{u|0<\Re u<2K, \Im u=0\big\},\\
\mathcal{N}_{-}&=\big\{u|0<\Re u<2K, \Im u=\pm K'\big\}, 
\end{align*}
where $K$ and $K'$ are elliptic integrals given in (\ref{K}) and (\ref{K'}). These cycles are the boundaries for the holomorphic factorization of the summability kernels.
Introduce the notation, $\lambda_{p}:=\lambda(e^{ip})=e^{\gamma(p)}$.  
We will prove the following theorem.
\begin{theorem}\label{V+V_lemma}
Near the curve, $$\mathcal{N}_{-}=\big\{u: 0<\Re u < 2K, \Im u=\pm K'\big\},$$
we have the identity
\begin{align}\label{VVV+}
\frac{\tilde{V}_{+}(u)}{z(u)^{2M+1}+1}=\frac{\tilde{V}_{-}(u)}{z(u)^{2M+1}-1}
\end{align} while near the curve, $$\mathcal{N}_{+}=\big\{u: 0<\Re u < 2K, \Im u=0\big\},$$
we have the identity
\begin{align}\label{VVV-}
\frac{\tilde{V}_{+}(u)}{z(u)^{2M+1}+1}=-\frac{\tilde{V}_{-}(u)}{z(u)^{2M+1}-1},
\end{align}
where \begin{align*}
\tilde{V}_{+}(u):=\sqrt{\frac{(\lambda(u)-\lambda_{\pi})}{(\lambda(u)-\lambda_{0})}}\frac{\prod_{p>0 \in \Sigma_{A}}(\lambda(u)-\lambda_{p})}{\prod_{p>0 \in \Sigma_{P}}(\lambda(u)-\lambda_{p})}
\end{align*}
\begin{align*}
\tilde{V}_{-}(u):=\sqrt{\frac{(\lambda(u)-\lambda^{-1}_{0})}{(\lambda(u)-\lambda^{-1}_{\pi})}}\frac{\prod_{p>0 \in \Sigma_{P}}(\lambda(u)-\lambda^{-1}_{p})}{\prod_{p>0 \in \Sigma_{A}}(\lambda(u)-\lambda^{-1}_{p})}.
\end{align*}
The square roots are here chosen to have positive real parts.
\end{theorem}
Before we give the proof we explain our principal concern with Theorem \ref{V+V_lemma}.
Recall from Appendix A that the two cycles $\mathcal{M}_{+}$ and $\mathcal{M}_{-}$ on the spectral curve $\mathcal{M}$, where $\lambda<1$ and $\lambda>1$ respectively, are given by
$$\mathcal{M}_{\pm}=\{(z,\lambda)=(e^{i\theta},e^{\mp\gamma(\theta)})\}.$$ 
 Recall from Appendix \ref{Appendix C} that in the uniformization parameter $u$, the cycles $\mathcal{M}_{\pm}$ are located at the following positions in the periodic parallelogram,
$$\mathcal{M}_{\pm}=\big\{u:0<\Re u < 2K, \Im u=\pm \tfrac{K'}{2}\big\}.$$ 
We observe that in a neighborhood of the cycle $\mathcal{M}_{+}$, the function $\tilde{V}_{+}(u)$ is a multiple of $V_{+}(u)$ in (\ref{V+}) from the Bugrij-Lisovyy formula  while $\tilde{V}_{-}(u)$ is a multiple of $V_{-}(u)$ in (\ref{V-}) in a neighborhood of the cycle  $\mathcal{M}_{-}$.
The factorization problem solved by $\tilde{V}_{\pm}$ makes it possible to transform sums over the periodic points to the anti-periodic points.
A related factorization in the scaling limit was employed in \cite{Lisovyy} to compute a relevant Green function for the Dirac operator on the cylinder. Subsequently, in \cite{PH10} the technique in \cite{Lisovyy} was used to compute the matrix elements of the spin operator for the periodic Ising model in the scaling limit. It is expected that this technique can be extended to the finite periodic lattice to provide an alternative proof of the Bugrij-Lisovyy conjecture. If one uses formulas for the spin matrix elements on a finite periodic lattice, one can control the convergence of the scaling limit in a much simpler way. This convergence result is given in \cite{Hy09} assuming the Bugrij-Lisovyy conjecture  \cite{BL03}.

We now return to the proof of Theorem \ref{V+V_lemma}.
\begin{proof}
Recall from Appendix \ref{Appendix C} that $0<2a<K'$ is defined by
$s_{1}=-i\sn(2ia).$ A simple calculation shows that a substitution of $a$ with $a+\frac{K'}{2}$ in the elliptic
parametrization of the Boltzmann weights, interchanges $z$ and $\lambda$ in the spectral curve (\ref{spectral}), and sends $s_{1}$ to $-s_{2}$, $s_{2}$ to $-s_{1}$, and $c_{1}c_{2}$ to $-c_{1}c_{2}$.
We therefore introduce the relation
$$ia':=\frac{iK'}{2}+ia.$$
Then 
\begin{align}\label{zu}
\lambda':
= \lambda(u+\tfrac{iK'}{2},a')=z(u,a)=k\sn(u+ia)\sn(u-ia)
\end{align} and
\begin{align}\label{z'}z':=z(u+\tfrac{iK'}{2},a')=\lambda(u,a)=\frac{\sn(u-ia)}{\sn(u+ia)}
\end{align} for $k=\frac{1}{s_{1}s_{2}}$.
Interchanging $\lambda$ and $z$ through the substitution, $a\mapsto a+\frac{iK'}{2}$, the spectral curve given in (\ref{spectral}) becomes
\begin{align}\label{spectralcurve}
s_{1}\frac{\lambda'+\lambda'^{-1}}{2}+s_{2}\frac{z'+z'^{-1}}{2}=c_{1}c_{2}.
\end{align}
Using the spectral curve in (\ref{spectralcurve}), we can write
\begin{align}\label{spectral1}
\prod_{p \in \Sigma_{P}}(c_{1}c_{2}-s_{1}\cos p-s_{2}\tfrac{z'+z'^{-1}}{2})& \nonumber=\prod_{p \in \Sigma_{p}}s_{1}(\tfrac{\lambda'+\lambda'^{-1}}{2})-s_{1}\cos p\\
& \nonumber= (\tfrac{s_{1}}{2})^{2M+1}\prod_{p \in \Sigma_{p}}((\lambda'-e^{ip})(1-\lambda'^{-1}e^{-ip}))\\
& = (\tfrac{s_{1}}{2})^{2M+1}({\lambda'}^{2M+1}-1)^{2}{\lambda'}^{-(2M+1)}
\end{align}
and similarly,
\begin{align}\label{spectral2}
\prod_{p \in \Sigma_{A}}(c_{1}c_{2}-s_{1}\cos p-s_{2}\tfrac{z'+z'^{-1}}{2})&= (\tfrac{s_{1}}{2})^{2M+1}({\lambda'}^{2M+1}+1)^{2}{\lambda'}^{-(2M+1)}.
\end{align}
Then from (\ref{zu}), (\ref{spectral1}) and (\ref{spectral2}) we obtain
\begin{align}\label{spectralnullvector}
&\nonumber\mathrel{\phantom{=}}\bigg(\frac{z(u)^{2M+1}+1}{z(u)^{2M+1}-1}\bigg)^{2}\\
&\nonumber=\bigg(\frac{{\lambda'}^{2M+1}+1}{{\lambda'}^{2M+1}-1}\bigg)^{2}\\
& =\frac{\prod_{p \in \Sigma_{A}}(c_{1}c_{2}-s_{1}\cos p-s_{2}\tfrac{z'+z'^{-1}}{2})}{\prod_{p \in \Sigma_{P}}(c_{1}c_{2}-s_{1}\cos p-s_{2}\tfrac{z'+z'^{-1}}{2})}.
\end{align}
Now we rewrite the factor in (\ref{spectralnullvector}) using (\ref{z'}) and the spectral curve in (\ref{spectral}).
It  becomes
\begin{align*} 
&\mathrel{\phantom{=}}\frac{\prod_{p \in \Sigma_{A}}(\lambda_{p}+\lambda^{-1}_{p}-(z'+z'^{-1}))}{\prod_{p \in \Sigma_{P}}(\lambda_{p}+\lambda^{-1}_{p}-(z'+z'^{-1}))}\\
& = \frac{\prod_{p \in \Sigma_{A}}(\lambda_{p}-z')(1-\lambda^{-1}_{p}z'^{-1})}{\prod_{p \in \Sigma_{P}}(\lambda_{p}-z')(1-\lambda^{-1}_{p}z'^{-1})}\\
& =  \frac{\prod_{p \in \Sigma_{A}}(\lambda(u)-\lambda_{p})(\lambda(u)-\lambda^{-1}_{p})}{\prod_{p \in \Sigma_{P}}(\lambda(u)-\lambda_{p})(\lambda(u)-\lambda^{-1}_{p})}.
\end{align*}
Since $\lambda(p)=\lambda(-p)$ for $p\neq 0,\pi$, the right hand side of the equation above can be written
\begin{align*} &\mathrel{\phantom{=}}\frac{(\lambda(u)-\lambda_{\pi})(\lambda(u)-\lambda^{-1}_{\pi})}{(\lambda(u)-\lambda_{0})(\lambda(u)-\lambda^{-1}_{0})}\frac{\prod_{p>0 \in \Sigma_{A}}(\lambda(u)-\lambda_{p})^{2}(\lambda(u)-\lambda^{-1}_{p})^{2}}{\prod_{p>0 \in \Sigma_{P}}(\lambda(u)-\lambda_{p})^{2}(\lambda(u)-\lambda^{-1}_{p})^{2}}.
\end{align*}
Now define
\begin{align*}
\tilde{V}_{+}(u):=\sqrt{\frac{(\lambda(u)-\lambda_{\pi})}{(\lambda(u)-\lambda_{0})}}\frac{\prod_{p>0 \in \Sigma_{A}}(\lambda(u)-\lambda_{p})}{\prod_{p>0 \in \Sigma_{P}}(\lambda(u)-\lambda_{p})},
\end{align*}
\begin{align*}
\tilde{V}_{-}(u):=\sqrt{\frac{(\lambda(u)-\lambda^{-1}_{0})}{(\lambda(u)-\lambda^{-1}_{\pi})}}\frac{\prod_{p>0 \in \Sigma_{P}}(\lambda(u)-\lambda^{-1}_{p})}{\prod_{p>0 \in \Sigma_{A}}(\lambda(u)-\lambda^{-1}_{p})}
\end{align*}
such that we have \begin{align}\label{VVVV}
\bigg(\frac{\tilde{V}_{+}(u)}{z(u)^{2M+1}+1}\bigg)^{2}=\bigg(\frac{\tilde{V}_{-}(u)}{z(u)^{2M+1}-1}\bigg)^{2}.
\end{align}
In order to figure out the correct signs near the curves $\mathcal{N}_{\pm}$ when we take the square root in (\ref{VVVV}), it is enough to check the identities in (\ref{VVV+}) and (\ref{VVV-}) at one point. We will check (\ref{VVV+}) by letting $u=K\pm iK'$ which are points on the curve $\mathcal{N}_{-}$. 
Using the identities (\ref{translation}), the addition formula (\ref{additionsn}), and Theorem \ref{palmer}, we obtain
\begin{align*}
&\phantom{=}z(K\pm iK')=k\sn(K\pm iK'+ia)\sn(K\pm iK'-ia)\\
& = k\bigg(\frac{k^{-1}\ns(K)\cn(ia)\dn(ia)-k^{-1}\ds(K)\cs(K)\sn(ia)}{1-\ns^{2}(K)\sn^{2}(ia)}\bigg)\times\\
& \times \bigg(\frac{k^{-1}\ns(K)\cn(ia)\dn(ia)+k^{-1}\ds(K)\cs(K)\sn(ia)}{1-\ns^{2}(K)\sn^{2}(ia)}\bigg).
\end{align*}
Here we also used the fact that $\sn(u)$ is an odd function of $u$ while $\cn(u)$ and $\dn(u)$ are even functions of $u$, \cite[p. 493]{WW82}.
Using the fact that $\sn(K)=1$, $\cn(K)=0$ \cite[p. 499]{WW82}, the identities in (\ref{cddnsn}) and the fact that $\alpha_{2}=-k^{-1}\ns^{2}(ia)$  \cite[p. 338]{palmer1}, we obtain
\begin{align*}
z(K\pm iK')&=\frac{k^{-1}\dn^{2}(ia)}{1-\sn^{2}(ia)}\\
& = \frac{k^{-1}(1-k^{2}\sn^{2}(ia))}{1-\sn^{2}(ia)}\\
& = \frac{k^{-1}(1+k\alpha^{-1}_{2})}{1+k^{-1}\alpha^{-1}_{2}}
\end{align*} which can be checked to be greater than one for $T<T_{C}$. Thus, $\frac{1}{z(K\pm iK')^{2M+1}-1}$ is positive. A similar calculation gives,
$\lambda(K\pm iK')=1$, and hence $\tilde{V}_{+}(u)$ and $\tilde{V}_{-}(u)$ are both positive for $u=K\pm iK'$. Thus the identity in (\ref{VVV+}) follows. Now we prove (\ref{VVV-}) by letting $u=K$ which is a point on the curve $\mathcal{N}_{+}$. Then it can be checked that 
$$z(K)=\frac{k\cn^{2}(ia)}{\dn^{2}(ia)}=z(K\pm iK')^{-1}<1,$$ and hence $\frac{1}{z(K)^{2M+1}-1}$ is negative. We also have $\lambda(K)=1$ so that $\tilde{V}_{\pm}(u)$ are positive for $u=K$. Hence (\ref{VVV-}) follows, and the theorem is proved.

\end{proof}

\thanks{The author would like to thank Professor John Palmer for his advice and support of this work. }

\appendix

\section{Grassmann Algebra and Fock Representations of the Clifford Algebra}\label{Appendix A}
In this section, we introduce the Fock representations of the Clifford algebra. We follow \cite{palmer1} closely and refer the reader to this work for more details.
Assume $W$ is a finite even-dimensional complex vector space with a distinguished nondegenerate complex bilinear form $(\cdot,\cdot)$.  
A subspace $V$ of $W$ is isotropic if $(x,y)=0$ for all $x,y\in V$.
We are interested in a decomposing of the space $W$ into two isotropic subspaces $W_{\pm}$ such that
\begin{align}\label{W}
W=W_{+}\oplus W_{-}.
\end{align} Such a splitting is called an isotropic splitting or a polarization.
We parametrize each isotropic splitting by an operator $Q$ defined by
\begin{align}\label{QAPx}
Qx= \left\{ \begin{array}{ll}
         x & \mbox{if $x \in W_{+}$};\\
        -x & \mbox{if $x \in W_{-}$}.\end{array} \right.
        \end{align}
Define $Q_{\pm}:=\tfrac{1}{2}(I\pm Q).$ Since $Q^{2}=I$, we observe that $Q_{\pm}$ are the projections onto the $\pm 1$ eigenspaces for $Q$. We notice that $Q_{+}+Q_{-}=I$ and $Q_{+}Q_{-}=0$. If we define $W_{\pm}:=Q_{\pm}W$, then the space $W$ is the direct sum given in (\ref{W}).
An operator $Q$ that is skew symmetric and satisfies $Q^{2}=I$ is called a polarization.
Let $S_{k}$ denote the group of permutations on the set $\{1,2,...,k\}$.
The linear operator defined by
$$W^{\otimes k}\ni w\mapsto \alt(w) :=\frac{1}{k !}\sum_{\sigma \in S_{k}}\sgn(\sigma)w_{\sigma_{1}}\otimes \cdot \cdot \cdot\otimes w_{\sigma_{k}}$$ is a projection from $W^{\otimes k}$ onto $\Alt^{k}(W)$, where $\Alt^{k}(W)$ is the space of alternating $k$ tensors over $W$.
The wedge product of $v\in \Alt^{k}(W)$ and $w\in \Alt^{l}(W)$ is here defined by
$$v\wedge w:=\frac{\sqrt{(k+l)!}}{\sqrt{k!}\sqrt{l!}}\alt(v \otimes w)\in \Alt^{k+l}(W)$$ and it follows that for $v_{i}\in W$ and $i=1,...,k$,
$$v_{1}\wedge v_{2}\wedge\cdot \cdot \cdot \wedge v_{k}=\frac{1}{\sqrt{k!}}\sum_{\sigma \in S_{k}}\sgn(\sigma)v_{\sigma_{1}}\otimes v_{\sigma_{2}}\otimes\cdot\cdot\cdot\otimes v_{\sigma_{k}}$$ (see \cite{palmer1} and \cite{Sp65}).
The Clifford algebra $\Cliff(W)$ of $W$ is the associative algebra with multiplicative unit $e$, generated by elements $x\in W$ that satisfy the Clifford relations,
\begin{align}\label{Clifford}
xy+yx=(x,y)e\quad \mbox{for}\quad x,y \in W.
\end{align}
For each polarization $Q$ of the isotropic splitting of the space $W$, there is a Fock representation $F_{Q}$ of the Clifford algebra $\Cliff(W)$. This representation acts on the alternating tensor algebra,
$$\Alt(W_{+}):=\bigoplus_{k=0}^{n}\Alt^{k}(W_{+}),$$ where $\Alt^{0}(W_{+})=\mathbb{C}$ and $n=\dim(W_{+})$.
The Fock representation is defined as
$$W\ni x \mapsto F_{Q}(x):=c(x_{+})+a(x_{-})$$ for $x_{\pm}=Q_{\pm}x$. Here $W_{-}$ is identified with the dual $W^{*}_{+}$ via the nondegenerate complex bilinear form $W_{+}\ni x_{+}\mapsto (x_{+},x_{-})$ for $x_{-}\in W_{-}$.
The creation operator $c(x_{+})$ associated with $x_{+}\in W_{+}$ acts on $\Alt^{k}(W_{+})$ in the following way, \begin{align}\label{creationop} \Alt^{k}(W_{+})\ni v \mapsto c(x_{+})v=x_{+}\wedge v \in \Alt^{k+1}(W_{+}).\end{align} The annihilation operator $a(x_{-})$ associated with $x_{-}\in W_{-}$ is defined as 
$a(x_{-}):=c^{\tau}(x_{-})$, where $c^{\tau}(x_{-})$ is the transpose of $c(x_{-})$ with respect to the complex bilinear form $(\cdot,\cdot).$ It is given by
\begin{align}\label{annihilationop} a({x_{-}})v=\sum_{j=1}^{k}(-1)^{j-1}({x_{-}},v_{j})v_{1}\wedge \cdot \cdot \cdot\wedge \hat{v}_{j}\wedge \cdot \cdot \cdot\wedge v_{k}\end{align} for $v=v_{1}\wedge \cdot \cdot \cdot\wedge v_{j}\wedge \cdot \cdot \cdot\wedge v_{k}\in \Alt^{k}(W_{+}),$ and
where $\hat{v}_{j}$ signifies that the factor $v_{j}$ is omitted from $v$.
It can be checked that the creation and annihilation operators satisfy the anti-commutation relations,
\begin{equation}\label{anticommutingA1}
\begin{array}{lccc}
c(x_{+})c(y_{+})+ c(y_{+})c(x_{+})=0,\\
a(x_{-})a(y_{-})+ a(y_{-})a(x_{-})=0,\\
a(x_{-})c(y_{+})+ c(y_{+})a(x_{-})=(x_{-},y_{+})I
\end{array}
\end{equation}
for $x_{\pm}, y_{\pm}\in W_{\pm}$.
Since $W_{\pm}$ are isotropic subspaces and by using the anti-commutation relations given in (\ref{anticommutingA1}), it is not hard to see that $F_{Q}$ satisfies the generator relations for the Clifford algebra,
\begin{align}\label{Fock1}
F_{Q}(x)F_{Q}(y)+F_{Q}(y)F_{Q}(x)=(x,y)I\quad \mbox{for} \quad x,y \in W.
\end{align} 
When the polarization is understood, we will drop the subscript $Q$ and write
$F:=F_{Q}$.
We write
$$0: =1\oplus 0\oplus \cdot \cdot \cdot\oplus 0$$ for the vacuum vector in $\Alt(W_{+})$. The vacuum vector is defined to be the unique vector that is annihilated by all the elements in $W_{-}$ in the $F_{Q}$ representation of the Clifford algebra. 

Define the Clifford group $\mathcal{G}$ as the group of invertible elements $g$ in the Clifford algebra $\Cliff(W)$ that satisfy
\begin{align*}
gvg^{-1}=T(g)v\quad \mbox{for}  \quad v \in W \subseteq \Cliff(W)\end{align*}
for some linear map $T(g)$ on $W$.
It follows from this equation that $T$ is complex orthogonal, i.e.
$$(T(g)v,T(g)w)=(v,w)\quad \mbox{for}\quad v,w \in W.$$

\section{Berezin Integral Representation for the Matrix Elements}\label{Appendix B}
In this section we introduce a representation of the creation and annihilation operators which is an analog of the holomorphic representations as given in Faddeev and Slavnov \cite{FS80}.
We will use this representation to write the matrix elements for the Fock representation of an element $g$ in the Clifford group as Pfaffians of a skew symmetric matrix whose entries depend on the matrix elements of the induced rotation associated with $g$.

Assume $W$ is a finite even-dimensional complex vector space with a Hermitian inner product $\langle \cdot, \cdot \rangle$ and a distinguished nondegenerate complex bilinear form $(\cdot, \cdot)$ defined by
$$(u,v)=\langle \bar{u},v \rangle \quad \mbox{for}\quad u,v \in W,$$
where $u\mapsto \bar{u}$ is a conjugation. We consider a Hermitian polarization,
$$W=W_{+}\oplus W_{-},$$ where $W_{\pm}$ are isotropic subspaces of $W$ as defined in Appendix A.
Let $\{e^{+}_{k}\}$ and $\{e^{-}_{k}\}$  denote  orthonormal bases for $W_{+}$ and $W_{-}$ with respect to the Hermitian symmetric inner product, and suppose that $\{e^{+}_{k}\}$ and $\{e^{-}_{k}\}$ are dual with respect to the complex bilinear form, i.e. $( e^{+}_{k},e^{-}_{l})=\delta_{kl}$.
Suppose that $W$ has dimension $2M$ and define
\begin{align}\label{ebasis}e^{\pm}_{I}:=e^{\pm}_{I_{1}}\wedge \cdot \cdot \cdot\wedge e^{\pm}_{I_{k}}\quad \mbox{for}\quad 1\leq I_{1}< \cdot \cdot \cdot< I_{k}\leq M.
\end{align}
The set $\{e^{\pm}_{I}\}$ is then an orthonormal basis for $\Alt(W_{\pm})$, where we define $e^{+}_{\emptyset}=1$.
Let $\mathcal{P}$ denote the collection of subsets of $\{1,.... ,M\}$.
For an element $J$ in $\mathcal{P}$, we write
$$J=\{J_{1}, J_{2},...,J_{k}\}\quad \mbox{with}\quad J_{1}<J_{2}<\cdot \cdot \cdot<J_{k}.$$
We write $\#J=k$ for the number of elements in $J$. If $R$ is a $2M\times 2M$ matrix, we let $R_{I,J}$ denote the $(\#I+\#J) \times (\#I+\#J)$ submatrix of $R$ made from the rows and columns of $R$ indexed by $I$ and $J$ respectively.
An element in $\Alt(W_{+})$ is given by
$$G(e^{+}):=\sum_{I\in \mathcal{P}}G_{I}e^{+}_{I},$$ where the map, $\mathcal{P}\ni I \to G_{I}\in \mathbb{C}$.
Here $G(e^{+})$ can be thought of as a ``polynomial'' in the elements $e^{+}_{I}$ in the exterior algebra.
For $1\leq i\leq M$, the creation operator $c(\cdot)$ acts on $\Alt^{k}(W_{+})$ as
$$\Alt^{k}(W_{+})\ni v \mapsto c(e^{+}_{i})v=e^{+}_{i}\wedge v \in \Alt^{k+1}(W_{+})$$ for $e^{+}_{i}\in W_{+}$. 
For $1\leq i\leq M$, the annihilation operator,
 $a(e^{+}_{i}):=\frac{\partial}{\partial e^{+}_{i}}$, is analogous to a ``derivative'', and is the linear map
$$\frac{\partial}{\partial e^{+}_{i}}:\Alt(W_{+})\to \Alt(W_{+})$$ defined as follows:
If the monomial, $X:={e^{+}_{{i}_{1}}}\wedge {e^{+}_{{i_{2}}}}\wedge \cdot \cdot \cdot\wedge e^{+}_{{i}_{n}}$, contains exactly one factor $e^{+}_{i}$ then
$$\frac{\partial}{\partial e^{+}_{i}}({e^{+}_{{i}_{1}}}\wedge {e^{+}_{{i_{2}}}}\wedge \cdot \cdot \cdot\wedge {e^{+}_{{i}_{n}}})=\pm{e^{+}_{{i}_{1}}}\wedge \cdot \cdot \cdot\wedge {\hat{e}^{+}_{{i}}}\wedge \cdot \cdot \cdot\wedge{e^{+}_{{i}_{n}}},$$
where ${\hat{e}^{+}_{i}}$ signifies that the factor ${e^{+}_{i}}$ is omitted from $X$, and the plus or minus sign is determined by number of interchanges the operator $\frac{\partial}{\partial e^{+}_{i}}$ has to make from the left before it contracts the factor ${e^{+}_{i}}$. An even number of interchanges gives a positive sign and an odd number gives a minus sign. 
If the monomial does not contain any factor $e^{+}_{i}$ or if it contains more than one factor of $e^{+}_{i}$, then
$\frac{\partial}{\partial e^{+}_{i}}X=0.$ 
We define Berezin integrals as linear functionals in the following way, $$\int e^{+}de^{+}=1, \quad \int e^{-}de^{-}=1, \quad \int de^{+}=0, \quad \int de^{-}=0,$$
where we assume that $de^{-}$ and $de^{+}$ anti-commute with each other as well as with $e^{-}$ and $e^{+}$. 
An element in the Grassmann algebra $\Alt(W)$ is given by
$$G(e^{+},e^{-})=G_{00}+G_{01}e^{-}_{1}+G_{10}e^{+}_{1}+G_{11}e^{-}_{1}\wedge e^{+}_{1}+...+G_{1,...,M,M,...,1}e^{-}_{1}\wedge..\wedge e^{-}_{M}\wedge e^{+}_{M}\wedge...\wedge e^{+}_{1},$$
where $G_{00}$, $G_{01}$, $G_{10}$ ,..., and $G_{1,...,M,M,...,1}$ are complex numbers \cite{FS80}.
The integral of $G(e^{+}, e^{-})$ is then defined as
$$\int G(e^{+}, e^{-})\prod_{k=1}^{M}de^{+}_{k}de^{-}_{k}=G_{1....M, M...1}.$$
For $G$ and $H$ in $\Alt(W_{+})$, we have  \cite[p. 53]{FS80}
\begin{align}\label{innerpb}
\langle G,H \rangle=\int \bar{G}(e^{+})H(e^{+})e^{-\sum_{k=1}^{M}e^{+}_{k}\wedge e^{-}_{k}}\prod_{k=1}^{M}de^{+}_{k}de^{-}_{k},
\end{align} where
$$\bar{G}(e^{+})=\sum_{I\in \mathcal{P}}\bar{G}_{I}e^{-}_{I}$$ and the conjugation is defined as
$$\overline{G_{I}e^{+}_{M}\wedge \cdot \cdot \cdot \wedge e^{+}_{1}}=\bar{G}_{I}e^{-}_{1}\wedge \cdot \cdot \cdot \wedge e^{-}_{M}.$$
It is here understood that $e^{\sum e^{+}_{k}\wedge e^{-}_{k}}$ is the power series in the exterior algebra,
$$\sum_{k=0}^{M}\bigg(\frac{\sum e^{+}_{k}\wedge e^{-}_{k}}{k!}\bigg)^{k}.$$
We are in particular interested in two polarizations,
$$W=W^{A}_{+}\oplus W^{A}_{-}\quad \mbox{and}\quad W=W^{P}_{+}\oplus W^{P}_{-}.$$ Here $W^{P}_{\pm}$ and $W^{A}_{\pm}$ are isotropic subspaces of $W$ defined by $$W^{A}_{\pm}=Q^{A}_{\pm}W \quad \mbox{and}\quad W^{P}_{\pm}=Q^{P}_{\pm}W,$$ where
$$Q^{A}_{\pm}:=\tfrac{1}{2}(I\pm Q^{A}), \quad Q^{P}_{\pm}:=\tfrac{1}{2}(I\pm Q^{P}),$$ and where $Q^{A}$ and $Q^{P}$ are polarizations as defined in (\ref{QAPx}).
Let $F^{P}$ and $F^{A}$ denote the Fock representations associated with the Clifford algebra $\Cliff(W)$ acting on $\Alt(W^{P}_{+})$ and  $\Alt(W^{A}_{+})$ respectively. 
We consider a map
$$g:\Alt(W^{P}_{+})\to \Alt(W^{A}_{+}),$$ which satisfies the intertwining relation $$gF^{P}(x)=F^{A}(T(g)x)g$$ for $ x\in W$, and where $T:=T(g)$ is the induced rotation associated with $g$. 
We write
\begin{align*}
T(g):=\left( \begin{array}{cc}
A & B  \\
C & D  \\
\end{array} \right)
\end{align*} for the matrix of the induced rotation for $g$, where $T$ is a map
 $W^{P}_{+}\oplus W^{P}_{-}\to W^{A}_{+}\oplus W^{A}_{-}$.
Since $T(g)$ is a complex orthogonal matrix, we have the relation
$$T(g)^{\tau}T(g)=I,$$ where
$$T^{\tau}=\left( \begin{array}{cc}
D^{\tau} & B^{\tau}  \\
C^{\tau} & A^{\tau}  \\
\end{array} \right).$$
This relation implies the following identities,
\begin{equation}\label{eq6}
\begin{array}{clcc}
D^{\tau}A+B^{\tau}C^{\tau} = I, \quad D^{\tau}B+B^{\tau}D=0\\
C^{\tau}A+A^{\tau}C = 0, \quad C^{\tau}B+A^{\tau}D =I
\end{array}
\end{equation}
when $D$ is invertible.
Introduce the anti-commuting ``dummy'' variables
$\alpha^{\pm}_{i}\in W^{P}_{\pm}$ which anti-commute with $e^{\pm}_{i}$ as well. The action of the operator $g$ on $G(\alpha^{+})\in \Alt(W^{P}_{+})$ is given by
 \cite[p. 55]{FS80}
\begin{align*}
(g G)(e^{+})=\int g(e^{+}, \alpha^{-})G(\alpha^{+})e^{-\sum_{k=1}^{M}\alpha^{+}_{k}\wedge\alpha^{-}_{k}}\prod_{k=1}^{M}d\alpha^{+}_{k}d\alpha^{-}_{k}.
\end{align*}
Let $\{e^{\pm}_{i}\}$ and $\{\alpha^{\pm}_{j}\}$ denote orthonormal bases for $W^{A}_{\pm}$ and $W^{P}_{\pm}$ respectively and for $I,J \in \mathcal{P}$ define $e^{\pm}_{I}$ and $\alpha^{\pm}_{J}$ as in (\ref{ebasis}).
We first prove the following lemma.
\begin{lemma}\label{LemmaB}
Let $\tilde{g}(e^{+},\alpha^{-}):=e^{\mathcal{R}}$, where
$$\mathcal{R}=\sum_{l,m}[(\tfrac{1}{2}a_{lm}e^{+}_{l}\wedge e^{+}_{m})+(b_{lm}e^{+}_{l}\wedge \alpha^{-}_{m})+(\tfrac{1}{2}c_{lm}\alpha^{-}_{l}\wedge \alpha^{-}_{m})]$$ and $a=BD^{-1}$, $b=D^{-\tau}$ and $c=D^{-1}C$, where 
$$T=\left( \begin{array}{cc}
A & B  \\
C & D  \\
\end{array} \right)$$ is the matrix of a complex orthogonal map, $W^{P}_{+}\oplus W^{P}_{-}\to W^{A}_{+}\oplus W^{A}_{-}$.
Then 
\begin{align}\label{gtilde}
\tilde{g}G(e^{+})=\int \tilde{g}(e^{+},\alpha^{-})G(\alpha^{+})e^{-\sum_{k=1}^{M}\alpha^{+}_{k}\wedge\alpha^{-}_{k}}\prod_{k=1}^{M}d\alpha^{+}_{k}d\alpha^{-}_{k}
\end{align} defines a linear map which satisfies the intertwining relation,:
\begin{align}\label{FPFAL}
\tilde{g}F^{P}(x)G=F^{A}(T x)\tilde{g}G
\end{align}
for $G\in \Alt(W^{P}_{+})$ and $x \in W$.
\end{lemma}
\begin{proof}
For $\alpha^{+}_{i}, G\in \Alt(W^{P}_{+})$, we have  
\begin{align}\label{integralFP-}
(\tilde{g}F^{P}(\alpha^{+}_{i})G)(e^{+})&=\int \tilde{g}(e^{+},\alpha^{-})\alpha^{+}_{i}G(\alpha^{+})e^{-\sum_{k=1}^{M}\alpha^{+}_{k}\wedge \alpha^{-}_{k}}\prod_{k=1}^{M}d\alpha^{+}_{k}d\alpha^{-}_{k}.
\end{align}
By the signed Leibniz rule, we have \begin{align*}
\frac{\partial}{\partial \alpha^{-}_{i}}e^{-\sum_{k=1}^{M}{\alpha^{+}_{k}\wedge \alpha^{-}_{k}}}=\alpha^{+}_{i}e^{-\sum_{k=1}^{M}{\alpha^{+}_{k}\wedge \alpha^{-}_{k}}}
\end{align*} and it follows that the integral in (\ref{integralFP-}) can be written
$$(\tilde{g}F^{P}(\alpha^{+}_{i}G))(e^{+})=-\int\frac{\partial}{\partial\alpha^{-}_{i}}[\tilde{g}(e^{+},\alpha^{-})]G(\alpha^{+})e^{-\sum_{k}\alpha^{+}_{k}\wedge\alpha^{-}_{k}}\prod_{k=1}^{M}d\alpha^{+}_{k}d\alpha^{-}_{k},$$
where
\begin{align*}
-\frac{\partial}{\partial \alpha^{-}_{i}}\tilde{g}(e^{+},\alpha^{-})&=\bigg(\sum_{l}D^{-\tau}_{li}e^{+}_{l}-\sum_{l}(D^{-1}C)_{il}\alpha^{-}_{l}\bigg)\tilde{g}(e^{+},\alpha^{-})\\
& = (D^{-\tau}+D^{-1}C)\alpha^{+}_{i}\tilde{g}(e^{+},\alpha^{-}).
\end{align*}
Now we have for $v \in \Alt(W^{P}_{+})$,
\begin{align*}
F^{A}(T\alpha^{+}_{i})v
 = \sum_{k=1}^{M}\bigg[A_{ki}e^{+}_{k}+C_{ki}\frac{\partial}{\partial e^{+}_{k}}\bigg]v.
\end{align*}
It follows that
\begin{align*}
&\mathrel{\phantom{=}}(F^{A}(T\alpha^{+}_{i})\tilde{g}G)(e^{+})\\
&=\int\bigg[ \sum_{k=1}^{M}A_{ki}e^{+}_{k}+C_{ki}\frac{\partial}{\partial e^{+}_{k}}\bigg]\tilde{g}(e^{+},\alpha^{-})G(\alpha^{+})e^{-\sum_{k=1}^{M}\alpha^{+}_{k}\wedge\alpha^{-}_{k}}\prod_{k=1}^{M}d\alpha^{+}_{k}d\alpha^{-}_{k}.
\end{align*}
We have
\begin{align*}
&\mathrel{\phantom{=}}\bigg[ \sum_{k}A_{ki}e^{+}_{k}+C_{ki}\frac{\partial}{\partial e^{+}_{k}}\bigg]\tilde{g}(e^{+},\alpha^{-})\\ 
& = \bigg(\sum_{k}A_{ki}e^{+}_{k}+\sum_{k,l}C_{ki}(BD^{-1})_{kl}e^{+}_{l}+\sum_{k,l}C_{ki}D^{-\tau}_{kl}\alpha^{-}_{l}\bigg)\tilde{g}(e^{+},\alpha^{-})\\
& = (A+D^{-\tau}B^{\tau}C+D^{-1}C)\alpha^{+}_{i}\tilde{g}(e^{+},\alpha^{-})\\
& = (D^{-\tau}+D^{-1}C)\alpha^{+}_{i}\tilde{g}(e^{+},\alpha^{-}),
\end{align*}
where we in the last equation used (\ref{eq6}). Thus, we have showed that $$\tilde{g}F^{P}(\alpha^{+}_{i})G=F^{A}(T\alpha^{+}_{i})\tilde{g}G$$ for $\alpha^{+}_{i},G \in \Alt(W^{P}_{+})$. We can prove (\ref{FPFAL}) in a similar fashion for $\alpha^{-}_{i}\in W^{P}_{-}$.
\end{proof}
Let $0_{P}$ and $0_{A}$ denote the unit vacuum vectors in $\Alt(W^{P}_{+})$ and $\Alt(W^{A}_{+})$ respectively.
Since the map $\tilde{g}$ satisfies the relation in (\ref{FPFAL}), the range of $\tilde{g}$ is invariant under the action of the Fock representation $F^{A}(x)$ for all $x\in W$. The Fock representation is irreducible \cite{palmer1}, so the range of $\tilde{g}$ is either trivial or all of $\Alt(W^{A}_{+})$. Since $\tilde{g}0_{P}$ is nonzero, the range of $\tilde{g}$ must be $\Alt(W^{A}_{+})$. Thus $\tilde{g}$ is invertible, so we have that $\tilde{g}$ is an element of the Clifford group.

Recall (see \cite{palmer1}) that the Pfaffian of a $2M\times 2M$ skew symmetric matrix $R$ with matrix elements $R_{j,k}$ is defined in the following way: Let $\{e_{j}\}$ denote the standard basis of $\mathbb{C}^{2n}$. The Pfaffian, $\Pr(R)$, of $R$ is defined by
\begin{align}\label{Pfaffian}
\frac{1}{2^{M}M!}\bigg(\sum_{j,k=1}^{2M}R_{j,k}e_{j}\wedge e_{k}\bigg)^{M}=\Pf(R)e_{1}\wedge \cdot \cdot \cdot\wedge e_{2M}.
\end{align}
We use Lemma \ref{LemmaB} to prove that $g$ can be written as an exponential, whose
argument is a quadratic form. 
\begin{theorem}\label{gPfaffian}
Suppose that $g$ satisfies the intertwining relation,
$$gF^{P}(x)v=F^{A}(Tx)gv\quad \mbox{for}\quad x\in W\quad \mbox{and}\quad v \in \Alt(W^{P}_{+}),$$ where $$T(g):=\left( \begin{array}{cc}
A & B  \\
C & D  \\
\end{array} \right)$$ is the matrix of its induced rotation, $T:W^{P}_{+}\oplus W^{P}_{-}\to W^{A}_{+}\oplus W^{A}_{-}$.
Suppose that the one-point function $\langle 0_{A}, g\, 0_{P}\rangle$ is nonzero.
Then the kernel $g(e^{+},\alpha^{-})$ of $g$ can be written as
$$g(e^{+},\alpha^{-})=\langle 0_{A}, g\, 0_{P} \rangle \sum_{I,J\in \mathcal{P}}\Pf(R_{I,J})e^{+}_{I}\wedge \alpha^{-}_{J},$$
where $\Pf(R_{I,J})$ is the Pfaffian of the $(\#I +\#J)\times(\#I +\#J)$ skew symmetric matrix,
\begin{align}\label{mR}R_{I,J}=\left( \begin{array}{cc}
BD^{-1}_{I\times I} & D^{-\tau}_{I\times J}  \\
-D^{-1}_{J\times I} & D^{-1}C_{J\times J}  \\
\end{array} \right).
\end{align}
The sum is over all such $I$ and $J$ with $\# I+\#J$ even.
\end{theorem}
\begin{proof}
Since $g$ satisfies the intertwining relation,
$$gF^{P}(x)=F^{A}(Tx)g\quad \mbox{for}\quad x\in W,$$ Lemma \ref{LemmaB} implies that there is a nonzero constant $\lambda$ such that
\begin{align}\label{g}
g(e^{+}, \alpha^{-})=\lambda\, e^{\mathcal{R}},
\end{align}
where
$$\mathcal{R}=\sum_{l,m}[(\tfrac{1}{2}a_{lm}e^{+}_{l}\wedge e^{+}_{m})+(b_{lm}e^{+}_{l}\wedge \alpha^{-}_{m})+(\tfrac{1}{2}c_{lm}\alpha^{-}_{l}\wedge \alpha^{-}_{m})]$$ for $a=BD^{-1}$, $b=D^{-\tau}$ and $c=D^{-1}C$.
It can be checked that
\begin{align*}
\langle 0_{A}, \tilde{g}0_{P} \rangle =\int e^{\tfrac{1}{2}\sum_{m,l=1}^{M}c_{lm}\alpha^{-}_{l}\wedge \alpha^{-}_{m}}e^{-\sum_{k=1}^{M}\alpha^{+}_{k}\wedge \alpha^{-}_{k}}\prod_{k=1}^{M}d\alpha^{+}_{k}\alpha^{-}_{k}.
\end{align*}
Writing out the Taylor series expansion in the integrand expression above, we obtain
\begin{align*}
\langle 0_{A}, \tilde{g}0_{P} \rangle =1.
\end{align*}
Hence, $$\lambda =\langle 0_{A}, g0_{P}\rangle$$ and
\begin{align}\label{e}
g(e^{+},\alpha^{-})=\langle 0_{A}, g0_{P}\rangle e^{\mathcal{R}}.
\end{align}
It is well-known that
\begin{align}\label{ePfaffian}
e^{\mathcal{R}}=\sum_{I,J\in \mathcal{P}}\Pf(R_{I,J})e^{+}_{I}\wedge \alpha^{-}_{J},
\end{align}
where $R_{I,J}$ is given in (\ref{mR}).
Here $\#I+\#J$ must be even in order to contribute to the sum.
\newline
Combining (\ref{e}) and (\ref{ePfaffian}), we obtain the result.
\end{proof}
\section{The Spectral Curve associated with the Induced Rotation for The Transfer Matrix}\label{Appendix C}
In this section, we recall facts about the Jacobian elliptic functions, $\sn(u,k)$, $\cn(u,k)$ and $\dn(u,k)$, where $u$ is the uniformization parameter and $k$ is the modulus. These functions play a central role in the holomorphic factorization of the ratio of summability kernels on the spectral curve in Section \ref{Spectral2}. They are also a key element in a product formula for the spin matrix elements on the finite periodic lattice as we will discover in Section \ref{Pfaffianformalism}.
We follow the introduction of the Jacobian elliptic functions as given in \cite{WW82} and \cite{palmer1}, and refer the reader to these books for more details. 
The spectral curve $\mathcal{M}$ associated with the induced rotation $T_{z}(V)$ for the transfer matrix is given by the set of $(z,\lambda)$ such that
\begin{align}\label{spectral}
s_{1}\frac{z+z^{-1}}{2}+s_{2}\frac{\lambda+\lambda^{-1}}{2}=c_{1}c_{2}.
\end{align}
The spectral curve is topologically a torus and the two fold covering, $(\lambda,z)\mapsto z$, is ramified at $z=\alpha^{\pm}_{1},\alpha^{\pm}_{2}$, where  \cite[p. 65]{palmer1} 
$$\alpha_{1}=(c^{*}_{1}-s^{*}_{1})(c_{2}+s_{2})\quad \mbox{and}\quad \alpha_{2}=(c^{*}_{1}+s^{*}_{1})(c_{2}+s_{2}).$$ The roots $\alpha_{1}$ and $\alpha_{2}$ were introduced in Section \ref{Ising} in connection with the Boltzmann weights.
We use the shorthand notation $\sn(u)$, $\cn(u)$ and $\dn(u)$, for $\sn( u,k)$, $\cn( u,k)$ and $\dn( u,k)$ since the modulus $k$ is fixed at $k=\frac{1}{s_{1}s_{2}}$ in our calculations.
The functions  $\sn(u)$, $\cn(u)$ and $\dn(u)$ are doubly periodic meromorphic functions of $u$ and they satisfy the equations
\begin{align}\label{cddnsn}
\sn^{2}(u)+\cn^{2}(u)=1, \quad \dn^{2}(u)+k^{2}\sn^{2}(u)=1
\end{align}
and
\begin{align}
\label{sn'}
\frac{d}{du}\sn(u)=\cn(u)\dn(u).
\end{align}
We use the standard notation \cite{palmer1},
$$\ns(u):=\frac{1}{\sn(u)}, \quad \cs(u):=\frac{\cn(u)}{\sn(u)}$$ and in general
$$\nx(u):=\frac{1}{\xn(u)}\quad \mbox{and}\quad \x\y(u):=\frac{\xn(u)}{\yn(u)},$$ where $\x$ and $\y$ are one of either $c$, $d$ or $s$.
The Jacobian elliptic functions satisfy the following addition formulas \cite[p. 496-497]{WW82},
\begin{align}
\sn(u+v)& = \frac{\sn (u)\cn (v) \dn (v)+\sn (v)\cn (u)\dn (u)}{1-k^{2}\sn^{2} (u)\sn^{2}(v)}, \label{additionsn} \\ 
\sn(u-v)& = \frac{\sn^{2}(u)-\sn^{2}(v)}{\sn(u)\cn(v)\dn(v)+\sn(v)\cn(u)\dn(u)}, \label{subtractionsn} \\ 
\cn(u+v)& = \frac{\cn (u)\cn (v) -\sn (u) \sn (v) \dn (u)\dn (v)}{1-k^{2}\sn^{2} (u)\sn^{2}(v)}, \label{additioncn} \\ 
\dn(u+v)& = \frac{\dn (u)\dn (v) -k^{2}\sn (u) \sn (v) \cn (u)\cn (v)}{1-k^{2}\sn^{2} (u)\sn^{2}(v)}. \label{additiondn} 
\end{align}
Introduce the elliptic integrals \cite[p. 501]{WW82},
\begin{align}\label{K}K=\int_{0}^{1}(1-t^{2})^{-\tfrac{1}{2}}(1-k^{2}t^{2})^{-\tfrac{1}{2}}\,dt,\end{align}
\begin{align}\label{K'}K'=\int_{0}^{1}(1-t^{2})^{-\tfrac{1}{2}}(1-k'^{2}t^{2})^{-\tfrac{1}{2}}\,dt\end{align}
for which the complementary modulus $k'$ is defined by 
$k^{2}+k'^{2}=1.$
If we translate the Jacobian elliptic functions by $iK'$, we have \cite[p. 503]{WW82}
\begin{align}\label{translation}
\begin{array}{ll}
\sn(u\pm iK')&=k^{-1}\ns(u),\\
\cn(u\pm iK')&=\mp ik^{-1}\ds(u),\\
\dn(u\pm iK')&=\mp i\cs(u).
\end{array}
\end{align}
Introduce the two cycles on $\mathcal{M}$,
$$\mathcal{M}_{\pm}=\{(z,\lambda)=(e^{i\theta},e^{\mp\gamma(\theta)})\},$$ for $\theta\in [-\pi, \pi)$.
Here we have introduced the notation $\gamma(\theta):=\gamma( e^{i\theta})$, where $\gamma$ is defined as the positive root of
\begin{align*}\cosh \gamma(z)=c^{*}_{2}c_{1}-s_{2}^{*}s_{1}\big(\tfrac{z+z^{-1}}{2}\big).\end{align*} 
Palmer \cite[p. 71]{palmer1} showed that an elliptic substitution gives a uniformization of the whole complex curve:
\begin{theorem}\cite{palmer1}\label{palmer}
The map,
$$[0,2K]\times i[-K',K']\ni u\mapsto (z(u,a),\lambda(u,a))$$
with
$$z(u,a)=k\sn(u+ia)\sn(u-ia)$$
and
$$\lambda(u,a)=\frac{\sn(u-ia)}{\sn(u+ia)}$$
is a uniformization of the spectral curve,
$$s_{1}\frac{z+z^{-1}}{2}+s_{2}\frac{\lambda+\lambda^{-1}}{2}=c_{1}c_{2},$$
where $k=\frac{1}{s_{1}s_{2}},$ and
$0<2a<K'$ is defined by
$$s_{1}=-i\sn(2ia).$$
In the uniformization parameter $u$, the cycles $\mathcal{M}_{\pm}$ are located at 
$$\mathcal{M}_{\pm}=\bigg\{u:0< \Re u< 2K, \Im u=\pm \frac{K'}{2}\bigg\}.$$
\end{theorem}

\end{document}